\documentclass[envcountsect]{llncs}
\usepackage{graphicx}
\usepackage[pdftex,dvipsnames]{xcolor}  
\definecolor{darkblue}{rgb}{0,0,0.5}
\usepackage{amsmath}
\usepackage{amsfonts}   
\usepackage{amssymb}    

\usepackage{amsthm}
\usepackage{microtype}        
\microtypesetup{
  final,        
  activate={true,nocompatibility},  
  kerning=true,
  spacing=true,
  protrusion=true
}
\microtypecontext{spacing=nonfrench}
\usepackage{tikz}             
\usetikzlibrary{decorations.markings}
\usetikzlibrary{positioning}
\usetikzlibrary{arrows}
\usetikzlibrary{calc,decorations.pathmorphing,patterns,fit}
\usepackage{braket}           
\usepackage{hyperref}         
\hypersetup{
 pdfencoding=auto,         
 bookmarksnumbered=true,   
 colorlinks=true,          
 linkcolor=darkblue,       
 citecolor=darkblue,       
 filecolor=darkblue,       
 urlcolor=darkblue,        
}
\usepackage{cleveref}
\usepackage{enumitem}
\usepackage[
n,
advantage,
operators,
sets,
adversary,
landau,
probability,
notions,
logic,
ff,
mm,
primitives,
events,
complexity,
asymptotics,
keys,
]{cryptocode}
\usepackage{pifont}
\usepackage{bbm}
\usepackage{tcolorbox}
\usepackage{afterpage}
\usepackage[textwidth=0.7in,textsize=footnotesize,disable]{todonotes}
\newcommandx{\celine}[2][1=]{\todo[inline,linecolor=blue,backgroundcolor=blue!25,bordercolor=blue,#1]{#2}}
\newcommandx{\celinein}[1]{\todo[fancyline,linecolor=blue,backgroundcolor=blue!25,bordercolor=blue]{#1}}
\newcommandx{\marc}[2][1=]{\todo[inline,linecolor=OliveGreen,backgroundcolor=OliveGreen!25,bordercolor=OliveGreen,#1]{#2}}
\newcommandx{\huyin}[2][1=]{\todo[fancyline,linecolor=Plum,backgroundcolor=Plum!25,bordercolor=Plum,#1]{#2}}
\newcommandx{\huy}[2][1=]{\todo[inline,linecolor=Plum,backgroundcolor=Plum!25,bordercolor=Plum,#1]{#2}}
\newcommand{\ketbra}[2]{{\vert#1\rangle\!\langle#2\vert}}
\DeclareMathAlphabet\mathbfcal{OMS}{cmsy}{b}{n}

\spnewtheorem{claim}{Claim}[section]{\bfseries}{\rmfamily}

\newcommand{\pad}{-0.5\belowdisplayskip}
\newcommand{\slightpad}{-0.3\belowdisplayskip}
\newcommand{\bigpad}{-\belowdisplayskip}

\newcommand{\distinguish}[2]{\stackrel{\mathclap{\normalfont\mbox{\scriptsize{#2}}}}{#1}}

\newcommand\iseq{\stackrel{\mathclap{\normalfont\mbox{\scriptsize{?}}}}{=}}
\renewcommand{\negl}[1]{\mathsf{negl}\left(#1\right)}
\newcommand{\intval}[1]{\left[ #1 \right]}

\newcommand{\ideal}[1]{\mathcal{F}_{\textnormal{\textsf{#1}}}}

\newcommand{\uccmd}[1]{\textsf{#1}}
\renewcommand{\adv}[1]{\mathcal{#1}}
\newcommand{\bases}{\{+, \times\}}

\newcommand{\fname}[1]{\mathtt{#1}}

\newcommand{\fsub}[1]{\mathtt{#1}}

\newcommand{\size}[1]{\left|#1\right|}
\newcommand{\pr}[1]{\textnormal{Pr}\left[#1\right]}
\renewcommand{\log}[1]{\textnormal{log}\left(#1\right)}
\renewcommand{\poly}[1]{\mathsf{poly}(#1)}

\newcommand{\hilbert}[1]{\mathcal{H}_{#1}}
\newcommand{\markovstate}[3]{\rho_{#1\leftrightarrow #2\leftrightarrow #3}}
\newcommand{\sminentro}[1]{\textnormal{H}_\infty^\varepsilon\left(#1\right)}
\newcommand{\smaxentro}[1]{\textnormal{H}_0^\varepsilon\left(#1\right)}
\newcommand{\entro}[2]{\textnormal{H}_{#1}\left(#2\right)}
\newcommand{\td}[2]{\delta\left(#1, #2\right)}
\newcommand{\dis}[1]{d\left(#1\right)}
\newcommand{\hamming}[1]{{d}_{H}\left(#1\right)}
\newcommand{\tr}[1]{\textnormal{tr}\left(#1\right)}

\newcommand{\etal}{\emph{et al}.}
\newcommand{\ie}{i.e.}
\newcommand{\eg}{e.g.}
\newcommand{\etc}{etc.}

\newcolumntype{P}[1]{>{\centering\arraybackslash}p{#1}}

\newcommand{\Commit}{\mathtt{Commit}}
\newcommand{\Verify}{\mathtt{Verify}}
\newcommand{\KeyGen}{\mathtt{KeyGen}}
\newcommand{\com}{\mathcal{E}}

\newcommand{\Setup}{\mathtt{Setup}}
\newcommand{\Sign}{\mathtt{Sign}}
\newcommand{\Verif}{\mathtt{Verify}}
\newcommand{\param}{\mathtt{param}}
\newcommand{\OSign}{\mathtt{OSign}}
\newcommand{\Exp}{\mathtt{Exp}}
\newcommand{\A}{\mathcal{A}}

\renewcommand{\secpar}{\lambda}
\newcommand{\seck}{\secpar}
\newcommand{\sets}{\leftarrow}

\newcommand{\comif}{\texttt{IF\ }}
\newcommand{\comor}{\texttt{OR\ }}
\newcommand{\comelse}{\texttt{ELSE\ }}
\newcommand{\comreturn}{\texttt{RETURN\ }}

\newcommand{\Succ}{\mathsf{Succ}}

\usepackage{changepage}
\makeatletter
\patchcmd{\@maketitle}{\begin{center}}{\begin{adjustwidth}{-0.5in}{-0.5in}\begin{center}}{}{}
\patchcmd{\@maketitle}{\end{center}}{\end{center}\end{adjustwidth}}{}{}
\makeatother

\begin{document}
%
\title{On the Security of Password-Authenticated\\[2mm]
Quantum Key Exchange}

%
%

\author{C\'eline Chevalier\thanks{CRED, Universit\'e Panth\'eon-Assas, Paris II, France, \texttt{celine.chevalier@ens.fr}} \and
Marc Kaplan\thanks{VeriQloud, France, \texttt{kaplan@veriqloud.fr}} \and
Quoc Huy Vu\thanks{DIENS, \'Ecole normale sup\'erieure, CNRS, INRIA, PSL University, Paris, France, \texttt{quoc.huy.vu@ens.fr}}}
\institute{}

%
\maketitle              

\begin{abstract}

Motivated by the Quantum Key Distribution (QKD) protocol, introduced in 1984 in
the seminal paper of Bennett and Brassard, we investigate in this paper the
achievability of unconditionally secure password-authenticated quantum key
exchange (quantum PAKE), where the authentication is implemented by the means
of human-memorable passwords. We first show a series of impossibility results
forbidding the achievement of very strong security, leaving open the
feasibility of achieving a weaker security notion.
We then answer this open question positively by presenting a
construction for quantum PAKE that provably achieves
everlasting security in the simulation-based model.
Everlasting security is a security notion introduced by
M{\"u}ller-Quade and Unruh in 2007, which implies unconditional security after
the execution of the protocol and only reduces the power of the adversary to
be computational during the execution of the protocol, which seems quite a
reasonable assumption for nowadays practical use-cases.

 \end{abstract}






%
%
%
\setcounter{footnote}{0}






\section{Introduction}

In their 1984 seminal paper~\cite{BEN84}, Bennett and Brassard gave the
first proof that the laws of quantum mechanics could lead to an
achievement of \emph{unconditional security} for classical cryptographic
tasks. Their celebrated Quantum Key Distribution protocol (so-called
QKD) allows two parties to agree on a common secret key which is
information-theoretic secret, assuming a quantum channel and an
authenticated (but not secret) classical channel.



Even though this protocol is a conceptual milestone in the quantum
cryptography field, the need for an information-theoretically authenticated
classical communication channel leads to a bootstrapping problem.
In practice, implementations of unconditionally secure
QKD leave no choice but requiring Alice and Bob to use a
pre-shared short random secret key (to authenticate the messages with
authentication codes constructed from universal hashing) in order to
obtain a larger random secret key.
Another unavoidable problem is that the authentication keys can be run
out, because either the adversary makes the execution fail (denial-of-service
attack) or due to technical problems (the parties cannot exclude that an
eavesdropper was in fact present).
Moreover, when considering large scale quantum networks, in which secure
communication should be possible between any pair of nodes,
the requirement for pre-shared randomness does not scale well: each node would
have to store a number of keys, which is linear in the size of the network,
let alone the problem of key management.




On the contrary, in so-called \emph{authenticated} key exchange, the two
parties are
able to generate a shared cryptographic secret key, to be later used with
symmetric primitives in order to protect communications, while interacting
over an \emph{insecure} network under the control of an adversary.
Various authentication means have been introduced for classical networks.
The most
practical ones are certainly based on either Public Key Infrastructures (PKI)
or human-memorable passwords.
The latter leads to PAKE, standing for \emph{Password-Authenticated Key
Exchange}.
PAKE protocols allow users to
securely establish a common cryptographic key over an \emph{insecure and
unauthenticated} channel only using a low-entropy, human-memorable secret key
called a \emph{password}.
The advantage of a PAKE, in sharp contrast to all QKD-like
schemes, is that no authenticated channel is needed.
In the classical setting, PAKE has been extensively
studied, resulting in various secure and efficient protocols. However,
classical PAKE protocols can only achieve computational security, where the
adversary's power is computationally limited. Thus, it is natural to ask the
following question:

\begin{quotation}
\begin{center}
\emph{Can we achieve a provably stronger security notion for
password-based key exchange protocols using quantum communication?}
\end{center}
\end{quotation}



Unfortunately, even if QKD raised a lot of hope on unconditional security
using quantum mechanics, a series of no-go theorems showed that the dream of
unconditional security brought by quantum communication will never be a
reality for many cryptographic tasks. For instance, several attempts have been
made to achieve unconditionally secure quantum bit-commitments, until Mayers
and Lo and Chau independently showed that statistically hiding and
binding quantum commitments are impossible~\cite{mayers1997unconditionally,lo1997quantum}.

The impossibility of quantum cryptography was further extended to
oblivious transfer (OT) by Lo~\cite{lo1997insecurity},
and finally extended to non-trivial two-party
computation protocols by Salvail \emph{et al.} and Buhrman \emph{et
al.}~\cite{AC:SalSchSot09,buhrman2012complete}.
In these papers, the authors show
that any non-trivial functionality leaks some information to the adversary,
and that the security for one party implies complete insecurity for the other.
Intuitively, the insecurity of two-party quantum protocols follows from the
fact that the
protocol itself allows parties to input a superposed state rather than a
classical one, and perform an appropriate measurement on the outcome state. At
the end of the protocol, one party can always gain more information on the
input of the other than that gained using any honest strategy.

Despite these impossibility results, we answer the above question
affirmatively. Noting that these impossibility results are only proven for
\emph{statistical} security, we remark that
	overcoming the impossibility results on PAKE in a quantum setting requires
some restriction on the adversary. One approach is to limit the adversary's
quantum memory as in the bounded quantum-storage model (BQSM)
\cite{FOCS:DFSS05}. Nevertheless, most of the quantum protocols in BQSM would
completely (and quite efficiently) break down in the case the assumption fails
to hold. Instead, we consider here another plausible approach by assuming
restrictions on the adversary's computational power. Following
M{\"u}ller-Quade and Unruh~\cite{TCC:QuaUnr07,C:Unruh13}, we consider the
notion of \emph{everlasting} security, where the adversary's power is
computationally
\emph{bounded} during the protocol execution and becomes computationally
\emph{unlimited} after the execution. In other words, everlasting security
assumes that, at the precise moment of the execution of the protocol, the
computational power of an adversary is limited \emph{and} that certain
mathematical problems are hard.
	This notion is justified by the fact that the computational power required
to break a cryptosystem might not exist \emph{now}, but could exist \emph{in
the future}, and that the protocols should also be protected after its
execution. In particular, everlasting security can ensure the security of
protocols executed today against future quantum computers, when they become
available.

Unfortunately, even in this weaker setting, some impossibility results
still hold, so that we first conduct a comprehensive review in different
settings: security models for composition
(simulation-based or stand-alone for sequential composition, universal
composability for universal composition), security definitions
(everlasting and statistical), and finally trusted setup assumptions
(none, standard ones such as a common reference string and strong ones such as
signature cards). We show that some settings do not suffer from
the impossibility results and manage to construct, in a simulation-based
model, an everlastingly secure quantum PAKE assuming a common reference string
as setup assumption. Our work builds upon QKD, where the authentication is
soly guaranteed by means of the password.

\subsubsection{Related Work.}



\paragraph{Security Models.} Definitions for security allowing composition are
usually based on the \emph{real-world/ideal-world} simulation paradigm in
so-called simulation-based models. The simplest one (sometimes called
\emph{stand-alone}) requires the composition to be only \emph{sequential} (it
requires that at any point, only \emph{one} protocol invocation be in
progress). Stronger and more complicated models allow for
\emph{self-concurrent} composition, or even \emph{arbitrary} composition.
In the classical setting, the two best known security models allowing for
arbitrary composition are the Universal Composability (UC) framework introduced
by Canetti~\cite{FOCS:Canetti01} and Abstract Cryptography introduced by
Maurer and Renner~\cite{Maurer11abstractcryptography}. A general quantum
simulation-based model with
a sequential composition theorem has been refined by
Fehr and Schaffner in \cite{TCC:FehSch09}. Quantum security models in the UC
style have been proposed by Ben-Or \emph{et al.} in \cite{TCC:BHLMO05} and
refined by Unruh in \cite{EC:Unruh10}. In this latter paper, Unruh also gives
a theoretical separation result between the quantum and classical setting by
showing that, in the quantum world, bit-commitments are complete for
statistically secure MPC, while it is not the case in the classical setting.

\paragraph{Everlasting Security.} The concept of everlasting UC-security was
first introduced by
M{\"u}ller-Quade and Unruh in \cite{TCC:QuaUnr07}, in which they construct a
(classically) everlasting UC-secure commitment protocol from certain strong
assumptions, so-called signature cards. Unruh studies in~\cite{C:Unruh13} the
everlasting security in the quantum UC model~\cite{EC:Unruh10} and further
extends impossibility results on everlastingly realizing cryptographic tasks
from standard trusted set-up assumptions such as CRS or PKI.



\paragraph{QKD.} Despite the apparent simplicity of Bennett and
Brassard's QKD protocol~\cite{BEN84}, the first complete composable
security proof of QKD was only given in the mid-2000's by Renner
\cite{quant-ph/0512258}.
This length of time between the protocol and the proof can
be explained by the inner difficulty of transposing the concepts of
classical cryptography to the quantum world. The universal composability
of QKD has been first studied by Ben-Or \emph{et al.}
in~\cite{TCC:BHLMO05}. A thorough state of the art of QKD's proofs can
be found in Tomamichel and Leverrier's
article~\cite{Tomamichel2017largelyself}. Mosca, Stebila and Ustaoglu
study in~\cite{PQCRYPTO:MosSteUst13} the security of QKD in the
classical authenticated key exchange framework, and give a proof of the
folklore theorem that QKD, when used with computationally secure
authentication (\eg{}, quantum-secure digital signatures), is
everlastingly secure (which they call long-term security).
In parallel, researchers have studied the closely-related
subject of the authentication of quantum channels, the latest works
being that of Fehr and Salvail~\cite{EC:FehSal17}, and
Portmann~\cite{EC:Portmann17}. This is a slightly different approach,
which also requires a shared secret key. The advantage is that the key
can be recycled: If the message arrived unaltered, it means that the key
is still secured. Furthermore, Portmann proved the composability of his
result in the Abstract Cryptography model.



\paragraph{PAKE.} The main approach to construct a UC-secure PAKE
protocol in the classical setting follows from the KOY-GL
paradigm~\cite{EC:KatOstYun01,EC:GenLin03}, first formalized by Canetti
\emph{et al.} in~\cite{EC:CHKLM05} and improved in order to obtain very
efficient results (see \cite{TCC:KatVai11,AC:ABBCP13,AC:BlaChe16} for
instance). It uses two building blocks: a CPA-secure encryption scheme
supporting smooth projective hashing (SPHF), and a CCA-secure encryption
scheme. Using different tools than SPHF, Jutla and Roy also proposed
very efficient UC-secure PAKE schemes~\cite{AC:JutRoy15,TCC:JutRoy18}.


Canetti \emph{et al.} proposed another approach in \cite{PKC:CDVW12} that
relies on oblivious transfer as the main cryptographic building block and
bypasses the ``projective hashing'' paradigm. Informally, they first construct
a secure protocol for randomized equality computation assuming an
authenticated channel and then apply the generic \emph{Split Authentication}
transformation of Barak \emph{et al.} \cite{JC:BCLPR11} to the protocol that
realizes the ``split'' version of that protocol. Split functionalities adapt
functionalities which assume authenticated channels to an unauthenticated
channels setting.

Although we are not aware of any quantum PAKE protocol, Damg{\aa}rd
\emph{et al.} proposed in~\cite{C:DFSS07b} two password-based
identification protocols in the bounded quantum storage model: Q-ID,
which is only secure against dishonest Alice or Bob, and Q-ID$^+$, which
is also secure against man-in-the-middle attacks. However, only Q-ID is
truly password-based; in Q-ID$^+$, Alice and Bob, in addition to the
password, also need to share a high-entropy key. On the negative side,
no quantum computing power at all is necessary to break the scheme, only
sufficient quantum storage, because the dishonest party could store all
the communicated qubits as they are, and measure them one by one in
either the computational or the Hadamard basis and completely break the
scheme. Subsequent works improve Q-ID schemes and prove their security
based on various uncertainty relations \cite{bouman2012all}, or in a
different security model, \eg{}, the computational security by using the
Commit-and-Open technique \cite{C:DFLSS09}.





\subsubsection{Our Contributions.} Our main contribution consists in
constructing a quantum PAKE protocol achieving an everlasting security notion
(and thus providing a password-authenticated variant of QKD). Towards this
goal, we conduct the following study:


\begin{itemize}

\item We first study and understand which security results are impossible and
which ones might be achievable for quantum-polynomial-time PAKE protocols
within different settings. We partially answer the question by showing that,
in the simulation-based model, statistically secure PAKE with explicit
authentication is impossible in the plain model. The question remains open for
statistical security with trusted setups and everlasting security without
trusted setups, and we answer it positively for everlasting security with a
trusted setup, by actually constructing an everlastingly secure PAKE in the
simulation-based model, given a CRS as a trusted setup. In the universal
composability model, we show that statistically or everlastingly secure PAKE
with explicit authentication is impossible with standard trusted setups
including CRS or PKI.

\item Second, as a side contribution, we improve the framework for the
simulation-based model proposed by Fehr and Schaffner in \cite{TCC:FehSch09}
by employing a single security definition, instead of separate definitions for
correctness and security for each party. Thus, it seems easier to deal with:
one can analyze protocols and prove their security by formally defining
simulation strategies. Our model is simple, expressive and simultaneously
enjoys a general sequential composition theorem. These results are given in
Section~\ref{sect:sim_based}. This extends the classical framework to the
quantum setting, and we give a definition of everlasting security in that
model.

\item Finally, using the ideas from the split authentication mechanism
proposed in~\cite{JC:BCLPR11} to get rid of
authenticated channels, we propose a quantum PAKE protocol which is indeed
everlastingly secure in the security model described above. Our construction
is inspired by the Commit-and-Open technique introduced in \cite{C:DFLSS09}.
Our work extends and improves on this result by showing that a stronger
security notion (namely everlasting security in the simulation-based model)
can be achieved. Lying at the core of our proof is a simulation strategy that
allows the simulator to change the output of the simulated adversary. In the
UC model (as opposed to the simulation-based model), the environment machine,
which is an interactive distinguisher, \emph{externally} interacts with the
adversary throughout the execution. One very important artifact of this
definition is that the simulator no longer has control over the output of the
simulated adversary. In fact, the adversary is completely controlled by the
environment. This is because the UC framework models the fact that the
real-world adversary may have additional information from the environment,
\eg{}, from other running instances of the protocol, or from other
concurrently running protocols as well. On the other hand, in the
simulation-based model, the adversary is \emph{internally} simulated by the
simulator. The simulated adversary outputs \emph{nothing}, and the simulator
is in charge of its output: it can apply any arbitrary function to the
prescribed input of the adversary. This is safe in the simulation-based model,
because the adversary is ``detached'' from the environment. By exploiting this
major difference, we show that our protocol is provably secure in the
simulation-based model. These results are given in
Section~\ref{sect:protocol}.

\end{itemize}




\section{Preliminaries}

\subsection{Notations}
For a set $I =\left\{ i_1, \dots, i_\ell \right\} \subseteq \left\{ 1, \dots,
n \right\}$ and a $n$-bit string $x \in \bin^n$, we write $x|_I \coloneqq
x_{i_1} \cdots x_{i_\ell}$. It is sometimes convenient that all substrings of
this form have the same length, irrespective of the actual size $\ell$ of the
index set $I$. Therefore, $x|_I$ is implicitly padded with
sufficiently many zeros. For $a, b \in \RR$, $\intval{a}$ denote the
closed integer interval $\left\{x \in \ZZ \mid 0 \leq x \leq a \right\} $,
and $\left( a, b \right) $ denote the open real interval $\left\{ x \in \RR
\mid a < x < b \right\}$.

The logarithms in this paper are with respect to base 2 and denoted by
$\log{\cdot} $. We write $h$ for the binary entropy function $h \left( \mu
\right) = -\left( \mu\log{\mu} + (1-\mu)\log{1-\mu} \right) $. The notation
$\negl{\lambda}$ denotes any function $f$ such that
$f(\lambda) = \lambda^{-\omega(1)}$, and $\poly{\lambda}$ denotes any function
$f$ such that $f(\lambda) = \mathcal{O}(\lambda^c)$ for some $c > 0$.
Let $\hamming{\cdot, \cdot}$ be the
Hamming distance, and let $r_H \left( \cdot, \cdot \right)$ denote the
relative Hamming distance between two strings, \ie{}, the Hamming distance
normalized by their length.

\subsection{Security Models}

Throughout this paper, we assume basic familiarity with multiparty
computation and associated security models, mainly the real world-ideal
world paradigm, either in the simulation-based
setting~\cite{JC:Canetti00,TCC:FehSch09} or the universal
composability framework~\cite{FOCS:Canetti01,EC:Unruh10}. We refer the
interested reader to Appendix~\ref{app:models} for a brief overview of
these models.

\subsection{Quantum Computation}

In this section, we give a very brief
introduction to the quantum notions we use in this paper,
we refer to \cite{quant-ph/0512258,nielsen2002quantum} for further
explanations.

\subsubsection{Systems and States.} For any positive integer $d \in
\mathbb{N}$, $\mathcal{H}_d$ stands for the complex Hilbert space of dimension
$d$. Sometimes, we omit the dimension and simply write $\mathcal{H}$. The
state of a quantum-mechanical system in $\mathcal{H}$ is described by a
\emph{density operator} $\rho$. A density operator $\rho$ is normalized with
respect to the trace norm ($\textnormal{tr}(\rho) = 1$), Hermitian ($\rho^* =
\rho$) and has no negative eigenvalues. $\mathcal{P}\left( \mathcal{H}
\right) $ denotes the set of all density operators for a system $\mathcal{H}$.
$\mathbbm{1}$ denotes the identity matrix. When it is normalized with the dimension, denoted by $\frac{1}{\textnormal{dim}\left(\hilbert{}\right)}\mathbbm{1}$, it represents the \emph{fully mixed state}.

A \emph{generalized measurement} on a system $A$ is a set of linear operators
$\left\{ M_A^x \right\}_{x\in \mathcal{X}}$ such that $\sum_{x\in \mathcal{X}}
{M_A^x} ^{\dagger} M_A^x  = \mathbbm{1}_A$. The probability $p_x$ of observing
outcome $x$ is $p_x = \tr{{M_A^x} ^{\dagger} M_A^x \rho}$.

A quantum state $\rho \in \mathcal{P}\left( \mathcal{H} \right) $ is called
\emph{pure} if it is of the form $\rho = \ketbra{\varphi}{\varphi}$ for a
(normalized) vector $\ket{\varphi} \in \mathcal{H}$. For a density matrix
$\rho_{AB} \in \mathcal{P}\left( \hilbert{A} \otimes \hilbert{B} \right) $ of
a composite quantum system $\hilbert{A} \otimes \hilbert{B}$, we write $\rho_B
= \textnormal{tr}_A\left( \rho_{AB} \right) $ for the state obtained by
tracing out system $\hilbert{A}$. We sometimes omit the index of the subspace
that is traced out if it is clear from the context.

The pair $\{\ket{0}_+, \ket{1}_+\}$ (also written as $\{\ket{0}, \ket{1}\}$)
denotes the computational or $+$-basis, the pair $\{\ket{0}_\times,
\ket{1}_\times\}$ (also written as $\{\ket{+}, \ket{-}\}$)
denotes the Hadamard or $\times$-basis, where
$\ket{0}_\times = (\ket{0} + \ket{1})/\sqrt{2}$ and $\ket{1}_\times = (\ket{0}
- \ket{1})/\sqrt{2}$. We write $\ket{x}_\theta = \ket{x_1}_{\theta_1} \otimes
\dots \otimes \ket{x_n}_{\theta_n}$ for the $n$-qubit state where string $x =
(x_1, \dots, x_n) \in \bin^n$ in encoded in bases $\theta = \{\theta_1, \dots,
\theta_n\} \in \{+, \times\}^n$.

We often consider cases where a quantum state
may depend on some classical random variable $X$. In that case the state is
described by the density matrix $\rho_E^x$ if and only if $X = x$. For an
observer who has access to the state but not $X$, the reduced state is
determined by the density matrix $\rho_E \coloneqq \sum\nolimits_{x}P_X
\left( x \right)\rho_E^x$, whereas the joint state, consisting of the
classical $X$ and the quantum register $E$ is described by the density matrix
$\rho_{XE} \coloneqq \sum\nolimits_{x}P_X \left( x \right)\ketbra{x}{x}
\otimes \rho_E^x $, where we understand $\left\{ \ket{x} \right\}_{x\in
\mathcal{X}}$ to be the computational basis of $\hilbert{X}$. Joint
states with such classical and quantum parts are called \emph{cq-states}. We
also write $\rho_X \coloneqq \sum\nolimits_{x}P_X \left( x
\right)\ketbra{x}{x}$ for the quantum representation of the classical random
variable $X$.

By $\td{\rho}{\sigma} \coloneqq \frac{1}{2}\left\|\rho - \sigma\right\|_1$, we
denote the \emph{trace distance} between two quantum states $\rho$ and
$\sigma$. We call two quantum states $\rho$ and $\sigma$
\emph{trace-indistinguishable}, denoted $\rho \approx_\varepsilon \sigma$, if
there is a negligible function $\varepsilon$ such that for a $\lambda \in
\mathbb{N}$, $\td{\rho}{\sigma} \leq \varepsilon = \negl{\lambda}$.
\begin{definition}
Let $\rho_{XB} \in \mathcal{P}\left( \hilbert{X}\otimes \hilbert{B} \right) $
be a cq-state classical on $\hilbert{X}$. The \emph{trace-distance from
uniform} of $\rho_{XB}$ given $B$ is defined by
\[\dis{\rho_{XB}|B} \coloneqq \frac{1}{2}\left\| \rho_{XB} - \frac{1}{\textnormal{dim}\left(\hilbert{X}\right)}\mathbbm{1} \otimes
\rho_{B}\right\|_1.\]
\end{definition}


\subsubsection{(Conditional) Smooth Entropies.} We briefly introduce the
notions of min- and max-entropy. For a bipartite cq-state $\rho_{XB} \in
\mathcal{P}\left( \hilbert{X}\otimes \hilbert{B} \right)$, we define
\[
p_{guess}\left( X|B \right)_{\rho} = \underset{\left\{ M_B^x \right\}}{\textnormal{sup }} \sum_{x \in \mathcal{X}} \pr{X = x}_{\rho}\tr{M_B^x\;\rho_{B|X=x} \left( M_B^x \right)^{\dagger} },
\]
where the optimization goes over all generalized measurements on $B$.
\begin{definition} Let $\rho = \rho_{XB}$ be a bipartite density operator. The
\emph{min-entropy} and \emph{max-entropy} of $A$ conditioned on $B$ is defined
as
\begin{align*}
\entro{\infty}{X|B}_{\rho} &\coloneqq -\log{p_{guess}\left( X|B \right)_{\rho}}, \\
\entro{0}{X|B}_{\rho} &\coloneqq -\entro{\infty}{X|C}_{\rho},
\end{align*}
where $\rho_{XBC}$ is any pure state with $\textnormal{tr}_C\left( \rho_{XBC}
\right) = \rho_{XB}$.
\end{definition}

\begin{definition} Let $\rho = \rho_{AB}$ be a bipartite density operator and
let $\varepsilon \geq 0$. The $\varepsilon$-smooth min- and max-entropy of
$A$ conditioned on $B$ is defined as
\begin{align*}
\sminentro{A|B}_{\rho} &\coloneqq \underset{\bar{\rho}}{\textnormal{sup}}\;\entro{\infty}{A|B}_{\bar{\rho}}, \\
\smaxentro{A|B}_{\rho} &\coloneqq \underset{\bar{\rho}}{\textnormal{inf}}\;\entro{0}{A|B}_{\bar{\rho}},
\end{align*}
where the supremum ranges over all density operator $\bar{\rho} =
\bar{\rho}_{AB}$ which are $\varepsilon$-close to $\rho$.
\end{definition}
We sometimes omit the subscript if the state $\rho$ is clear from the context.

\subsubsection{Privacy Amplification. }
Recall that a class $\mathcal{F}_n$ of hash functions from $\bin^n$ to
$\bin^\ell$ is called \emph{two-universal}, if for any $x \neq y \in \bin^n$
and for $F$ uniformly chosen from $\mathcal{F}_n$, the collision probability
$\pr{F(x) = F(y)}$ is upper bounded by $1/2^{\ell}$. We recall the
quantum-privacy-amplification theorem of
\cite{TCC:RenKon05} as formulated in
\cite[Corollary 5.6.1]{quant-ph/0512258}.

\begin{theorem} \label{theorem:privacy_amplification}
Let $\rho_{XB} \in \mathcal{P}\left( \mathcal{H}_X \otimes \mathcal{H}_B
\right) $ be a cq-state classical on $\hilbert{X}$, let $\mathcal{F}$ be a
family of two-universal hash functions from $\mathcal{X}$ to $\bin^\ell$, and
let $\varepsilon > 0$. Then,
\[
\dis{\rho_{F(X)BF}|BF} \leq
2\varepsilon + 2^{-\frac{1}{2}\left( \sminentro{X|B}_{\rho}-\ell \right) },
\]
for $\rho_{F(X)BF} \in \mathcal{H}\left( \mathcal{H}_Z\otimes \mathcal{H}_B
\otimes \mathcal{H}_F \right) $ defined by $\rho_{F(X)BF} \coloneqq \sum_{f\in
\mathcal{F}}P_{F}\left( f \right) \rho_{f(X)B} \otimes \ketbra{f}{f}$.
\end{theorem}

\subsubsection{Private Error Correction.}
Finally, we recall the private error correction technique introduced in
\cite{STOC:DodSmi05} and generalized to the quantum setting in
\cite{TCC:FehSch08}. This tool allows to correct a constant fraction of
errors, by using a family of efficiently decodable linear codes, where the
syndrome of a string is close to uniform if the string has enough min-entropy
and the code is chosen at random from the family. Specifically, they show that
for every $0 <\lambda < 1$, there exists a \emph{$\delta$-biased} (as defined
in \cite{STOC:DodSmi05}) family $\mathcal{C} = \left\{ C_i \right\}_{i \in
\mathcal{I}}$ of $\left[ n, k, d \right]_2$-codes with $\delta < 2^{-\lambda
n/2}$.

The following theorem, which is a variant of Theorem 3.2 in
\cite{TCC:FehSch08}, establishes the closeness of the syndrome of a string $X$
to random, given a random index $i$ and any $q$-qubit state that may depend on
$X$.
\begin{theorem}
\label{theorem:private_ecc}
Let the density matrix $\rho_{XB} \in \mathcal{P}\left( \hilbert{X} \otimes
\hilbert{B} \right) $ be a cq-state classical on $\hilbert{X}$ with
$X \in \bin^n$. For any constant $0 < \lambda < 1$, let
$\left\{ C_i \right\}_{i \in \mathcal{I}}$ be a $\delta$-biased family of
random variables over $\bin^n$ having square bias $\delta^2 < 2^{-\lambda n}$,
and let $I$ be uniformly and independently
distributed over $\mathcal{I}$. Then
\[ \dis{\rho_{\left( C_I \oplus X \right)BI}|BI} \leq \delta \times
2^{-\frac{1}{2}\left( \sminentro{X|B}_{\rho} -n \right)}. \]
\end{theorem}

\begin{proof}
The original theorem in \cite{TCC:FehSch08} states for $\entro{2}{X|B}$. By
using Jensen’s inequality on R\'enyi entropy and means of smoothing,
our theorem follows immediately.
\end{proof}

\subsection{Cryptographic Primitives}

We assume basic familiarity with signatures schemes, denoted as $\sigma 
= (\Setup,\allowbreak \KeyGen,\allowbreak \Sign,\allowbreak \Verif)$, 
which are strongly existentially unforgeable under a quantum 
chosen-message attack, and with commitment schemes, more precisely 
dual-mode commitment schemes, denoted as $\com = (\KeyGen_\fsub{H}, 
\KeyGen_\fsub{B}, \Commit, \Verify,\allowbreak \fname{Open}, 
\fname{Ext})$, where $H$ stands for \emph{hiding} keys and $B$ for 
\emph{binding} keys.
Definitions can be found in Appendix~\ref{app:primitives}.


\section{On the Feasibility of Securely Realizing PAKE}
\label{sect:feasibility}\label{sec:impossibility}
	In this section, we show negative results on the achievable
security of Password-based Key Exchange protocols when allowed to use
quantum communication. We focus on two composability settings: Either a
``minimal'' simulation-based security following a real world-ideal world
paradigm, as defined in~\cite{TCC:FehSch09,JC:Canetti00}, or the
full universally composable security~\cite{FOCS:Canetti01,EC:Unruh10}.



Following the literature, we call \emph{plain} model the setting in
which there are no setup assumptions (such as public-key infrastructure
(PKI), common reference string (CRS), random oracles (ROM), \etc{}).
Following for instance \cite{STOC:KusLinRab06}, in which the authors
study the connections between information-theoretic security and
security under composition, we consider here the information-theoretic
setting, in which the adversary is polynomially unbounded. Informally,
the output of a real execution of the protocol with a real adversary
must be (perfectly or statistically) the same as the output of an ideal
execution with a trusted party and an ideal-world adversary/simulator.
	On the contrary, in the computational setting, we focus on the notion of
\emph{everlasting} security~\cite{TCC:QuaUnr07,C:Unruh13}, which
informally means that the adversary is polynomially bounded during the
execution of the protocol, and unbounded afterwards. This models an
adversary possibly saving transcripts today, in order to potentially use
them at the time a quantum computer is built.

\subsection{Implicit or Explicit Authentication}
\label{subsec:auth}
	We recall an important property of a PAKE protocol: it
guarantees that if the same password was entered, the generated session
key is the same for both parties, but they might not know at the end of
the protocol whether it is so. This property is known as \emph{implicit}
authentication, as opposed to \emph{explicit} authentication, in which
the parties know whether they share the same session key at the end of
the protocol. In both cases, the protocol should guarantee that if the
passwords were different, the session keys are independent and random.


The line of work for impossibility results that we continue here focuses
on \emph{non-trivial} protocols\footnote{As explained for instance in
\cite[Section~7]{EC:CHKLM05}, the results are only interesting for what
they call \emph{non-trivial} protocols, in which two parties agree on a
shared secret key at the end of the execution of the protocol (except
perhaps with negligible probability), if 1) they use the same password
and 2) the adversary passes all messages between the parties without
modifying them or inserting any messages of its own. This is required
since otherwise the \emph{empty} protocol in which parties do nothing
would securely realize any PAKE functionality.} with explicit
authentication. It is known at least since~\cite[Section
5]{EC:BelPoiRog00} that explicit authentication can be added at no
security cost to any protocol with implicit authentication, using a
\emph{key confirmation} technique. The obtained key~$K$ would be used as
the key for a PRF secure for 3~queries, one of the players would send
$PRF_K(1)$ to the other, the other would send $PRK_K(2)$ to the first
one, and both would end up using $PRF_K(0)$ as the final session
key\footnote{A trivial construction of such a (perfect) PRF would be to
split the key into three parts, use the two first parts as key
confirmations and the last one as the real session key.}. This implies
that the following results also hold for protocols with \emph{implicit}
authentication.


\subsection{Impossibility in the Simulation-Based Model}
\label{subsec:impossibility:SB}

\begin{theorem}
\label{theorem:no_sim_pake}
There is no statistically simulation-based secure PAKE
protocol with explicit authentication in the plain model.
\end{theorem}

To the best of our knowledge, no equivalent result is known for everlasting
security or when allowing setup assumptions, such as a common reference
string.

This theorem is proven in~\Cref{subsec:impossibility-proofs:SB}.

\subsection{Impossibility in the Universally Composability Model}
\label{subsec:impossibility:UC}

As in the classical case (Canetti \etal{} prove in~\cite{EC:CHKLM05} the
impossibility of universally composable PAKE in the plain model), the
(im)possibility of PAKE depends on the existence of some setup assumption.
As shown by Unruh in~\cite{C:Unruh13}, the classical notion of passive
adversaries (which copy all data) does not make sense in the quantum case.
He thus considers only \emph{unitary} protocols, which perform no
measurements (any protocol can be transformed into such a protocol using
additional quantum memory). Unruh then defines a functionality~$\mathcal{F}$ to
be \emph{quantum-passively-realizable} it there exists a unitary
protocol that realizes~$\mathcal{F}$ with respect to passive unlimited
adversaries (that follow the protocol exactly and do not even copy
information). The following lemma gives examples of
quantum-passively-realizable functionalities.

\begin{lemma}[{\cite[Lemma 8]{C:Unruh13}}] The following functionalities are
quantum-passively-realizable: $\mathcal{F}_{CT}$ (coin-toss),
$\mathcal{F}_{CRS}$ (common reference string), $\mathcal{F}_{EPR}$
(predistributed EPR pair), $\mathcal{F}_{PKI}$ (public key infrastructure;
assuming that the secret key is uniquely determined by the public key).
\end{lemma}


We state the following impossibility theorem, proven
in~\Cref{subsec:impossibility-proofs:UC}.

\begin{theorem} \label{theorem:no_uc_pake} There is no statistically or
everlastingly quantum-UC-secure PAKE protocol with explicit authentication
which only uses quantum-passively-realizable functionalities as trusted setup
assumptions.
\end{theorem}


\subsection{Avoiding Impossibility Results}

In summary, we have shown that, in the simulation-based model,
statistically secure PAKE with explicit authentication is impossible in
the plain model. The question remains open for statistical security with
a trusted setup, or for everlasting security with or without trusted
setups. In the following, we partially solve these open questions, by
actually constructing an everlastingly secure PAKE in the
simulation-based model, given a CRS as a trusted setup.

In the universal composability model, statistically or everlastingly
secure PAKE with explicit authentication is impossible with
quantum-passively-realizable functionalities as trusted setups. Unruh
shows in~\cite{C:Unruh13} that it is possible using signature cards as a
trusted setup (he even shows that this setup assumption is indeed complete for
everlastingly secure two-party computation).

\section{Definition of Security}
\label{sect:sim_based}
\subsection{Description of the Simulation-based Model}

Our definition follows the framework based on the
\emph{real-world/ideal-world} simulation paradigm put forward in
\cite{TCC:FehSch09} and enjoys sequential composition. The main features of
our model are that it is formally sound, simple and expressive, benefits from
a simpler security definition tailored to various assumptions on the
adversary's computational power.

Since we are interested in two-party quantum computations, we formalize the
real and ideal model executing the task with two parties and a \emph{static}
adversary who can control an arbitrary but fixed corrupted party. We only
consider either the setting where one of the parties is corrupted, or the
setting where none of the parties is corrupted, in which case the adversary
seeing the transcript between the parties should learn nothing.

\subsubsection{Execution in the ideal model.} Denote the participating parties
by
$P_1$ and $P_2$ and let $i \in \{1, 2\}$ denote the index of the corrupted
party, controlled by an adversary $\adv{A}$. An ideal execution for an ideal
functionality $\ideal{}$ proceeds as follows:
\begin{itemize}[align=left]
\item[\textbf{Inputs:}] We fix an arbitrary distribution $P_U$ for $P_1$'s
input, $P_V$ for $P_2$'s input. For honest $P_1$ and $P_2$, we assume the
common input state $\rho_{UV}$ to be classical, \ie{} of the form $\rho_{UV} =
\sum\nolimits_{u, v}P_{UV}(u,v)\ketbra{u}{u}\otimes \ketbra{v}{v}$ for some
probability distribution $P_{UV}$. The adversary $\adv{A}$ also has an
auxiliary classical input denoted by $Z$ as well as a quantum state $T'$ which
only depends on $Z$, such that for any honest player's input $W$ and his
classical ``side information'' $S$: $\rho_{SWZT'} = \markovstate{SW}{Z}{T'}$.
All parties are initialized with the same value $1^\lambda$ on their security
parameter tape (including the trusted party).
\item[\textbf{Send inputs to trusted party:}] The honest party $P_j$ sends its
prescribed input to the trusted party. The corrupted party $P_i$ controlled by
$\adv{A}$ may either abort (by replacing the input with a special
$\textsf{abort}_i$ message), send its prescribed input, or send some other
input of the same length to the trusted party by applying some completely
positive trace-preserving (CPTP) map. This decision is made by $\adv{A}$ and
may depend on its auxiliary input and the input value of $P_i$. Denote the
common input state sent to the trusted party by $\rho_{SU'ZV'}$. Upon receipt
of input from the parties, the trusted party measures the inputs in the
computational basis.
\item[\textbf{Early abort option:}] If the trusted party receives an input of
the form $\textsf{abort}_i$ for some $i \in \{1, 2\}$, it sends
$\textsf{abort}_i$ to the honest party $P_j$ and the ideal execution
terminates. Otherwise, the execution proceeds to the next step.
\item[\textbf{Trusted party sends output to adversary:}] At this point the
trusted party computes $\left( X, Y \right) = \left( id_S \otimes \mathcal{F}
\right)\rho_{SU'ZV'}$ and let $f_1 = (S, X)$ and $f_2 = (Z, Y)$ and sends
$f_i$ to party $P_i$ (\ie{} it sends the corrupted party its output).
\item[\textbf{Adversary instructs trusted party to continue or halt:}]
$\adv{A}$ sends either $\textsf{continue}$ or $\textsf{abort}_i$ to the
trusted party. If it sends $\textsf{continue}$, the trusted party sends $f_j$
to the honest party $P_j$. Otherwise, if $\adv{A}$ sends $\textsf{abort}_i$,
the trusted party sends $\textsf{abort}_i$ to party $P_j$.
\item[\textbf{Outputs:}] The honest party always outputs the output value it
obtained from the trusted party. The corrupted party outputs nothing. The
adversary $\adv{A}$ outputs any arbitrary CPTP map of the prescribed input of
the corrupted party, the auxiliary classical input $Z$, and the value $f_i$
obtained from the trusted party.
\end{itemize}

The $\textsf{ideal execution of } \mathcal{F}$, denoted by
$\textnormal{IDEAL}_{\mathcal{F}, \mathcal{A}}\left(\rho_{SUZV},
\lambda\right)$, is defined as the overall output state (augmented with honest
inputs) of the honest party and the adversary $\adv{A}$ from the above ideal
execution.

\subsubsection{Execution in the real model.} We next consider the real model
in
which a real two-party quantum protocol $\Pi$ is executed with no trusted
parties. In this case, the adversary $\adv{A}$ sends all messages in place of
the corrupted party, and may follow an arbitrary strategy. In contrast, the
honest party follows the instructions of $\Pi$. We consider a simple network
setting where the protocol proceeds in rounds, where in each round one party
sends a message to the other party.

Let $\ideal{}$ be as above and let $\Pi$ be a two-party quantum protocol for
computing $\ideal{}$. When $P_1$ and $P_2$ are both honest, we fix an
arbitrary joint probability distribution $P_{UV}$ for the inputs $U$ and $V$,
resulting in a common output state $\rho_{UVXYE} = \rho_E \otimes \Pi(U, V)$
with a well defined joint probability distribution $P_{UVXY}$, where $E$ is
the adversary's quantum system. For an honest $P_j$ and a dishonest $P_i$ who
takes as input a classical $Z$ and a quantum state $V'$ and output (the same)
$Z$ and a quantum state $Y'$, then the resulting overall output state
(augmented with the honest party's input $S$ and $U$) is $\rho_{SUXZY'} =
\left(id_{SU} \otimes \Pi\right)\rho_{SUUZV'}$.

The $\textsf{real execution of } \Pi$, denoted by $\textnormal{REAL}_{\Pi,
\mathcal{S}}\left(\rho_{SUZV}, \lambda\right)$, is defined as the overall
output state of the honest party and the adversary $\adv{A}$ from the real
execution of $\Pi$.



\begin{definition}
A two-party quantum protocol $\Pi$ is said to \emph{statistically
$\varepsilon$-securely emulate} an ideal classical functionality $\ideal{}$
\emph{with abort in the presence of static malicious adversaries} if for every
(possibly unbounded) adversary $\adv{A}$ for the real model, there exists an
(possibly unbounded) adversary (called the \emph{simulator}) $\adv{S}$ for the
ideal model, such that
\[
\textnormal{IDEAL}_{\mathcal{F}, \mathcal{S}}\left(\rho_{SUZV}, \lambda\right) \approx_\varepsilon \textnormal{REAL}_{\Pi, \mathcal{A}}\left(\rho_{SUZV}, \lambda\right),
\]
where $S, Z \in \bin^*$ and $\lambda \in \mathbb{N}$.
\end{definition}

We also give here an adapted definition of everlasting security in the
simulation-based paradigm. The execution in the ideal model and the real model
stays the same as for unconditional security, but we require that the real-world
adversary and ideal-world adversary are computationally bounded.

\begin{definition}
\label{def:2pc_everlasting}
A two-party quantum protocol $\Pi$ is said to \emph{everlastingly
$\varepsilon$-securely emulate} an ideal classical functionality $\ideal{}$
\emph{with abort in the presence of static malicious adversaries} if for every
quantum-polynomial-time adversary $\adv{A}$ for the real model, there exists a
quantum-polynomial-time adversary (called the \emph{simulator}) $\adv{S}$ for
the ideal model, such that
\[
\textnormal{IDEAL}_{\mathcal{F}, \mathcal{S}}\left(\rho_{SUZV}, \lambda\right) \approx_\varepsilon \textnormal{REAL}_{\Pi, \mathcal{A}}\left(\rho_{SUZV}, \lambda\right),
\]
where $S, Z \in \bin^*$ and $\lambda \in \mathbb{N}$.
\end{definition}

\subsection{Split Authentication: From Passive Security to
Active Security}
\label{sect:split_auth}


A common approach in designing multi-party quantum cryptographic protocols is
to treat the \emph{authenticated communication} aspect of the problem as
extraneous to the actual protocol design. That is, the adversary is assumed to
be unable to send classical messages in the name of uncorrupted parties, or
modify classical messages that the uncorrupted parties send to each other.
This means that authentication must be provided by some mechanism that is
external to the protocol itself, such as classical authenticated channels, as
in QKD.

On the contrary, it makes no sense to rely on authenticated channels for
realizing authenticated key-exchange, such as PAKE. But in the absence of such
strong authentication mechanisms, honest parties cannot distinguish the case
in which they interact with each other from the case in which they interact
with the adversary, so that the adversary can always partition the players and
engage in separate executions of the protocols with each of them, playing the
role of the other player.

To overcome this difficulty, our approach is to follow the \emph{Split
Authentication} transformation of \cite{JC:BCLPR11}: We consider a completely
unauthenticated setting, where all \emph{classical} messages sent by the
parties may be tampered with and modified by the adversary without the
uncorrupted parties being able to detect this fact. Then we modify the
protocol as described on \Cref{fig:compiler}: We add an extra first flow in
which the players exchange public verification keys for a signature scheme,
and check these values by exchanging signatures on these keys. Each classical
flow of the subsequent protocol is then signed using the associated private
signing key, and verified by the other player, who aborts in case it does not
match.

This transformation implies that the \emph{only} attack that the adversary can
carry out is to completely ``disconnect'' the two uncorrupted parties (during
the added first flow), and engage in completely separate executions with each
one of the two parties, where in each execution the adversary plays the role
of the other party. Intuitively, the transformation guarantees that the
adversary is limited to pursuing one of the two following strategies:

\begin{enumerate}

\item \emph{Passive attacks:} In this strategy, the adversary does not tamper
with the first flow, so that it can only carry out active attacks on the
quantum part of the channel, but it cannot carry active attacks on the
classical channel without being caught.

\item \emph{Independent executions:} In this strategy, the adversary
intercepts the first flows between the parties and engages in independent,
separate executions with each of them. We note that, in our simulation-based
model, the adversary can only run one execution at any point. Then, the
security is exactly the same as in the case where one of the parties is
corrupted.

\end{enumerate}

\afterpage{
\begin{figure}[!h]
\footnotesize
\begin{tcolorbox}[colback=white,size=small,sharp corners,colframe=black,before
upper={\parindent15pt}] {}

\noindent\textit{Link Initialization}

Upon activation, $P_b$ does the followings:
\begin{enumerate}
\item $P_b$ chooses a key pair $(\textsf{sk}_b, \textsf{vk}_b)$ for the
signature scheme.
\item $P_b$ sends $\textsf{vk}_b$ to $P_{1-b}$. (Recall that in an
unauthenticated network, sending $m$ to $P_i$ only means that the message is
given to the adversary.)
\item $P_b$ waits until it receives a key from $P_{1-b}$ (Recall that this key
is actually received from the adversary and does not necessarily correspond to
the key sent by the other party.) and defines $\textsf{sid}_b =
\langle\left(P_0, \textsf{vk}_0 \right),\left(P_1,
\textsf{vk}_1\right)\rangle$.
\item $P_b$ computes $\sigma_b = \fname{Sign}\left( \textsf{sid}_b \right)$
and sends $\left(\textsf{sid}_b, \sigma_b\right)$ to $P_{1-b}$.
\item $P_b$ waits until it receives a message $\left(\textsf{sid},
\sigma\right)$ from $P_{1-b}$. Then it checks that
$\fname{Verify}_{\textsf{vk}_{1-b}}(\textsf{sid}, \sigma) = 1$ and that
$\textsf{sid}_b = \textsf{sid}$. If all of these check pass, then $P_b$
outputs $\textsf{sid}_b$.
\end{enumerate}
\textit{Core Protocol Execution}
\begin{enumerate}
\item $P_b$ initializes a counter $c$ to zero.
\item $P_b$ runs protocol $\Pi$ according to the protocol specification, with
some slight modifications for classical messages authentication:
\begin{enumerate}
\item When $P_b$ wants to send a classical message $m$ to $P_{1-b}$, it signs
on $m$ together with $\textsf{sid}_b$, the recipient identity, and the counter
value. That is, $P_b$ computes $\sigma =
\fname{Sign}_{\textsf{sk}_b}(\textsf{sid}_i,m, P_{1-b}, c)$, sends $\left(P_b,
m, c, \sigma\right)$ to $P_{1-b}$ and increments $c$.
\item Upon receiving a message $\left(P, m, c, \sigma\right)$ allegedly from
$P_{1-b}$, $P_b$ first verifies that $c$ did not appear in a message received
from $P_{1-b}$ in the past. It then verifies that $\sigma$ is a valid
signature on $(\textsf{sid}_b, m, P, c)$ using the verification key
$\textsf{vk}_{1-b}$. If the verification fails, it halts.
\end{enumerate}
\end{enumerate}
\end{tcolorbox}
\caption{Compiled protocol $\mathcal{C}^{\varepsilon_{sig}}(\Pi)$.}
\label{fig:compiler}
\end{figure}
}

\begin{theorem}
\label{theorem:compiler}
Assume the existence of signature schemes that are existentially unforgeable
under an adaptive quantum chosen message attack (see definitions
in~\ref{app:primitives}).
Let $\Pi$ be a two-party quantum protocol that is everlastingly secure in
the authenticated-channel setting.
Then, the compiled protocol $\mathcal{C}^{\varepsilon_{sig}}(\Pi)$, resulting
by applying the transformation given in Figure \ref{fig:compiler},
is everlastingly secure against static,
malicious adversaries, according to
Definition \ref{def:2pc_everlasting}, with no authenticated channels.
\end{theorem}

The proof works almost the same as the proof given in \cite{JC:BCLPR11}, and is
given in Appendix \ref{app:proof_compiler}.


\section{Our Protocol}
\label{sect:protocol}


\subsubsection{High-Level Description.}
	We use the split authentication mechanism given in
\Cref{sect:split_auth}, so that we focus on the ``inner'' protocol
construction, which is a quantum PAKE assuming authenticated classical
channels (which means that the adversary is assumed to be unable to modify
classical messages sent by the uncorrupted parties). Applying the
transformation described in \Cref{fig:compiler} (using digital signatures)
will thus lead to a quantum PAKE where the authentication between two honest
parties is solely guaranteed by the password.




The full description of our PAKE protocol is provided in Figure
\ref{fig:main_protocol} and its schematic diagram is given in Figure
\ref{fig:main_protocol_diagram}.
From a high point of view, it starts with a preparation phase, in which the
client samples
random binary strings $x$ and~$\theta$, and sends the encoded quantum state of
$x$ using basis $\theta$. Next, a parameter estimation phase is done by means
of a dual-mode commitment scheme, which can be either perfectly hiding or
perfectly binding, depending on the chosen commitment key (see details in
\Cref{sec:dual-mode}). The main
difference
between the security of a PAKE protocol and QKD is the need to consider the
cases where one of the parties is corrupted.
Two-party quantum protocols can easily be broken by the
\emph{adversary purification} attack: the dishonest party can purify his
actions at the expense of additional quantum memory, and delay the
measurements until the other party reveals her chosen basis at a later stage,
and learn more information than what he was supposed to.
In order to enforce honest behavior, we use
the Commit-and-Open compiler formally introduced in \cite{C:DFLSS09},
and apply it to both parties.
This forces both parties to measure by asking them to commit to all the basis
choices and measurement results, and open some of them later.


After the estimation phase, both parties exchange a one-time pad
of their password encrypted using the chosen random basis.
We show that the session keys of both parties at the end are
random and independent for any pair of different passwords.

Finally, the post-processing phase consists, as QKD, of
error correcting and privacy amplification.
A new problem lies, however, in the error correcting step: to correct the errors
caused by either the adversary or the imperfection of the quantum channel, one
party may send a syndrome of the generated secret key to allow the other
party to recover the same key from its noisy version.
However, the syndrome may give extra information to a dishonest party.
To circumvent this problem, we employ the $\delta$-biased linear binary codes
introduced in \cite{STOC:DodSmi05}, which has an additional property
that the syndrome of a string with high min-entropy is close to uniform.


\subsubsection{Notations and Building Blocks.}
Let $\lambda$ denote the security parameter and let $k = \poly{\lambda}$
and some $\alpha \in \left( 0, 1 \right)$. Assume that both parties
share some password $pw \in \mathcal{D} \subseteq \bin{}^m$. We denote:
\begin{itemize}
\item $\mathfrak{c}: \mathcal{D}\rightarrow \left \{ +, \times \right \}^n$
the encoding function of a binary code of length $n$ with $m = \left |
\mathcal{D} \right|$ codewords and minimal distance $d$.
$\mathfrak{c}$ is chosen such that $n$ is linear in $\log{m}$ or larger,
and $d$ is linear in $n$, i.e. $d \coloneqq \gamma n$, for some constant
$\gamma$.
\item $\mathcal{F}$ a strongly two-universal class of hash functions from
$\bin{}^\ell$ to $\bin{}^\lambda$ for some parameter $\ell = n/2$.
\item $\left\{ \mathtt{synd}_j \right\}_{j\in \mathcal{J}}$ the
family of syndrome functions corresponding to a $\delta$-biased family
$\mathcal{C} = \left\{ C_j \right\}_{j\in \mathcal{J}}$ of linear error
correcting codes of size $\ell = n/2$, where $\delta < 2^{-\beta n/4}$,
for some constant $0 < \beta < 1$.
Let $\left\{ \fname{decode}_j \right\}_{j\in \mathcal{J}}$
be the corresponding decoding function. A random $C_j$ allows to efficiently
correct a $\tau$-fraction of errors for some constant $\tau > 0$.
\item $\com$ a dual-mode proof commitment scheme, and we denote $c
\leftarrow \com.\Commit\left( m \right) $ an execution of the commit phase
of a message $m$ (with some randomness).
We assume that the opening phase consists in the sender sending $m$ (and some
randomness used in the commit phase) and the receiver verifying via a
deterministic function $\com.\Verify\left(c, m \right)$.
\end{itemize}

\subsubsection{Security Result.}

\begin{figure}[!t]
\footnotesize
\begin{tcolorbox}[colback=white,size=small,sharp corners,colframe=black,before
upper={\parindent15pt}] {}

\noindent\textbf{Common reference string:} A pair $\left( ck, ck' \right) $,
which are the two perfectly hiding commitment keys for the dual-mode proof
commitment scheme $\com$.

\noindent\textbf{Protocol Steps:}
\begin{enumerate}
\item Upon activation, $P_1$ chooses $x \sample \bin^k$, $\theta \sample
\bases^k$, encodes each data bit from $x$ according to the corresponding basis
bit $\theta$, let the encoded state be $\ket{\Psi}$ and sends
$\mathtt{flow-zero} = \ket{\Psi}$ to $P_2$.
\item Upon receipt of $\mathtt{flow-zero}$ from $P_1$, $P_2$ chooses
$\hat{\theta} \sample \bases^k$, measures $\ket{\Psi}$ in basis $\hat{\theta}$ to get a classical
string $\hat{x}$. For each pair of bits of $\hat{x}$ and $\hat{\theta}$, it
uses $ck$ to commit $(c^0_i, c^1_i) = (\com.\Commit(\hat{x}_i),
\com.\Commit(\hat{\theta}_{i}))$, then sends $\mathtt{flow-one}_1 =
\left\{(c^0_i, c^1_i) \right\}_{i \in \intval{k}} $ to $P_1$.
\item Upon receipt of $\mathtt{flow-one}_1$ from $P_2$, $P_1$ chooses a random
subset $T_1 \subset_R \{1, \dots, k\}$ such that $\size{T_1} = \alpha k$, and
sends $\mathtt{flow-one}_2 =  T_1 $ to $P_2$.
\item Upon receipt of $\mathtt{flow-one}_2$ from $P_1$, $P_2$ opens all the
commitments restricted to the indices $i \in T_1$ and sends
$\mathtt{flow-one}_3 = \left\{\hat{x}_i, \hat{\theta}_i \mid i \in T_1
\right\} $ to $P_1$.
\item Upon receipt of $\mathtt{flow-one}_3$ from $P_2$, $P_1$ verifies all the
commitments restricted to the indices $i \in T_1$: $\com.\Verify ( c^0_i,
\hat{x}_i ) \iseq{} 1$ and $\com.\Verify ( c^1_i, \hat{\theta}_i ) \iseq{}
1$. Furthermore, it sets $T_1^\prime = \left\{ i \in T_1 \mid \theta_i =
\hat{\theta}_i \right\}$ and verifies that $r_H\left( x|_{T_1^\prime},
\hat{x}|_{T_1^\prime}\right) \leq \frac{\tau}{2} $. It aborts if the
verifications fail.

For each pair of bits of $x$ and $\theta$ restricted to the set $\left\{ 1,
\dots, k \right\} \setminus T_1$, it uses $ck'$ to commit $(c^0_i, c^1_i) =
(\com.\Commit(x_i), \com.\Commit(\theta_{i}))$, then sends $\mathtt{flow-two}_1
= \left\{ (c^0_i, c^1_i) \right\}_{i \in \intval{k_1}} $ to $P_2$, where $k_1
= k - \alpha k$.
\item Upon receipt of $\mathtt{flow-two}_1$ from $P_1$, $P_2$ chooses a random
subset $T_2 \subset_R \{1, \dots, k\} \setminus T_1$ such that $\size{T_2} =
\alpha k$, and send $\mathtt{flow-two}_2 = T_2 $ to $P_1$.
\item Upon receipt of $\mathtt{flow-two}_2$ from $P_2$, $P_1$ opens all the
commitment restricted to the indices $i \in T_2$ and sends
$\mathtt{flow-two}_3 = \left\{x_i, \theta_i \mid i \in T_2 \right\}$ to $P_2$.
\item Upon receipt of $\mathtt{flow-two}_3$ from $P_1$, $P_2$ verifies all the
commitment restricted to the indices $i \in T_2$: $\com.\Verify ( c^0_i,
x_i ) \iseq{} 1$ and $\com.\Verify ( c^1_i, \theta_i ) \iseq{} 1$. It
also sets $T_2^\prime = \left\{ i \in T_2 \mid \theta_i = \hat{\theta}_i
\right\}$ and verifies that $r_H\left( x|_{T_2^\prime},
\hat{x}|_{T_2^\prime}\right) \leq \frac{\tau}{2} $. It aborts if the
verifications fail.
\item Both parties set $\bar{T} = \left \{1, \dots, k \right \} \setminus
\left (T_1 \cup T_2  \right )$. $P_1$ computes $\varphi = \theta|_{\bar{T}}
\oplus
\mathfrak{c}(pw) $, and $P_2$ computes $\hat{\varphi} = \hat{\theta}|_{\bar{T}}
\oplus \mathfrak{c}(pw)$. $P_2$ sends $\mathtt{flow-three}_1 =
\hat{\varphi}$ to $P_1$.
\item Upon receipt of $\mathtt{flow-three}_1$ from $P_2$, $P_1$ sends
$\mathtt{flow-three}_2 = \varphi$ to $P_2$.
\item Both parties set $I_w = \{i \mid \varphi_i = \hat{\varphi}_i \}$. $P_1$
chooses $j \in_R \mathcal{J}$, computes $s = \mathtt{synd}_j \left( x|_{I_w}
\right) $ and sends $\mathtt{flow-four} = \left\{ j, s \right\}$ to $P_2$.
\item Upon receipt of $\mathtt{flow-four}$ from $P_1$, $P_2$ recovers
$\tilde{x}|_{I_w}$ from $\hat{x}|_{I_w}$ with the help of $s$, chooses $f
\in_R \mathcal{F}$ for privacy amplification, computes the session key
$\textsf{sk} = f \left( \tilde{x}|_{I_w} \right) $, sends $\mathtt{flow-five}
= f$ to $P_1$, outputs $\textsf{sk}$ and halts.
\item Upon receipt of $\mathtt{flow-five}$ from $P_2$, $P_1$ computes the
session key $\textsf{sk} = f \left( x|_{I_w} \right) $, outputs $\textsf{sk}$
and halts.
\end{enumerate}
\end{tcolorbox}
\caption{Protocol description.}
\label{fig:main_protocol}
\end{figure}

\begin{figure}[!t]
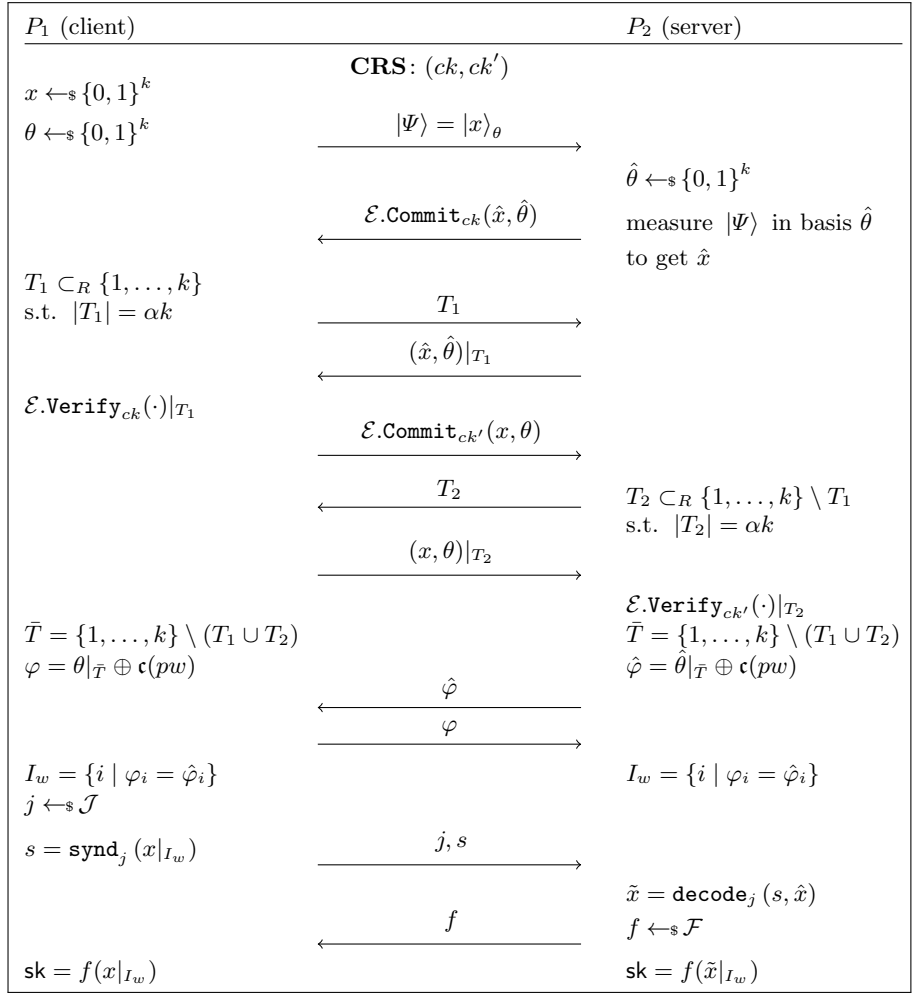

\centering
\fbox{\pseudocode{
P_1 \textnormal{ (client)} \<\< P_2 \textnormal{ (server)}
\\[0.1\baselineskip][\hline] \\[-0.5\baselineskip]
\< \qquad \mathbf{CRS\colon} (ck, ck') \\[\pad]
x \sample \bin^k \\[\pad]
\theta \sample \bin^k
\< \sendmessageright{top={$\ket{\Psi} = \ket{x}_\theta$}} \< \\[\pad]
\<\< \hat{\theta} \sample \bin^k \\[\pad]
\< \sendmessageleft*{\com.\Commit_{ck}(\hat{x}, \hat{\theta})} \<
\text{measure } \ket{\Psi} \text{ in basis } \hat{\theta} \\[\bigpad]
	\<\<\text{to get } \hat{x} \\[\pad]
T_1 \subset_R \{1, \dots, k\} \\[\bigpad]
	\text{s.t. } \size{T_1} = \alpha k \<
\sendmessageright{top={$T_1$}} \\[\bigpad]
\< \sendmessageleft*{(\hat{x}, \hat{\theta})|_{T_1}} \\[\pad]
\com.\Verify_{ck}(\cdot)|_{T_1} \\[\bigpad]
\< \sendmessageright*{\com.\Commit_{ck'}( x, \theta )} \\[\pad]
\< \sendmessageleft{top={$T_2$}} \< T_2 \subset_R \{1, \dots, k\} \setminus
T_1 \\[\bigpad]
	\<\<\text{s.t. } \size{T_2} = \alpha k \\[\bigpad]
\< \sendmessageright*{({x}, {\theta})|_{T_2}} \\[\pad]
\<\< \com.\Verify_{ck'}(\cdot)|_{T_2} \\[\slightpad]
\bar{T} = \left \{1, \dots, k \right \} \setminus \left (T_1 \cup T_2  \right
) \<\< \bar{T} = \left \{1, \dots, k \right \} \setminus \left (T_1 \cup T_2
\right ) \\[\pad]
\varphi = \theta|_{\bar{T}} \oplus \mathfrak{c}(pw) \<\< \hat{\varphi} =
\hat{\theta}|_{\bar{T}} \oplus \mathfrak{c}(pw)\\[\bigpad]
\< \sendmessageleft{top={$\hat{\varphi}$}} \\[\bigpad]
\< \sendmessageright{top={$\varphi$}} \\[\pad]
I_w = \{i \mid \varphi_i = \hat{\varphi}_i \} \<\< I_w = \{i \mid \varphi_i =
\hat{\varphi}_i \} \\[\slightpad]
j \sample \mathcal{J}\\[\pad]
s = \fname{synd}_j \left(x|_{I_w}\right) \< \sendmessageright*{j, s} \\[\pad]
\<\< \tilde{x} = \fname{decode}_j \left(s, \hat{x}\right) \\[\bigpad]
\< \sendmessageleft*{f} \< f \sample \mathcal{F} \\[\pad]
\textsf{sk} = f(x|_{I_w}) \<\< \textsf{sk} = f(\tilde{x}|_{I_w})}}
\caption{Schematic Diagram of the Protocol.}
\label{fig:main_protocol_diagram}
\end{figure}

\begin{theorem}
\label{theorem:main_theorem_1}
The protocol $\Pi$ of Figure \ref{fig:main_protocol} is everlastingly secure,
in the $\ideal{CRS}$-hybrid model, assuming authenticated classical channels.
\end{theorem}
\begin{proof}[Proof Sketch]$ $\newline
\emph{No corrupted parties.} We highlight some commons and differences
between our protocol and QKD protocol as follows.
\begin{itemize}
\item \texttt{flow-zero} is identical to QKD's preparation phase.
\item \texttt{flow-one} and \texttt{flow-two} correspond to QKD's parameter
estimation phase. \texttt{flow-three} corresponds to QKD's sifting phase. The
main differences lie in these steps. First, the order is inverse in QKD
protocol. We note that this ordering of steps makes no differences since the
qubits with different bases are not used in the protocol at all. Second,
parameter estimation is done by using the commitment scheme. Since the
commitment is perfectly hiding, it essentially gives the adversary nothing.
Third, instead of directly publishing the bases as in QKD, both parties
exchange one-time pads of their password. Again, because of the perfect
security property, this difference has no effect.
\item \texttt{flow-four} and \texttt{flow-five} correspond to QKD's
post-processing phase: error correcting and privacy amplification,
respectively.
\end{itemize}
Our security definition in this case shares a common ``picture''
with the one of QKD. Thus, we follow the main steps of QKD's proof \cite
{quant-ph/0512258,Tomamichel2017largelyself} with some modifications.
Particularly, we leverage a technical lemma from \cite{C:DFLSS09} to prove
statistical bounds on the min-entropy and max-entropy. Unlike QKD's
proofs, we also need to show that the password is independent of the
adversary's system. First note that after the commit-and-open phase, we are
close to the case where for \emph{any} choice of $T_1$ and $T_2$, and for
\emph{any} outcome $x_T$ when measuring $\ket{\Psi}$ in basis $\theta|_{T_2}$
and $\hat{\theta}|_{T_1}$, the relative error $r_H \left( x|_{T^\prime},
\hat{x}|_{T^\prime} \right)$ (where $T^\prime =
T_1^\prime \cup T_2^\prime$) gives an upper bound
on the relative error $r_H \left(x|_{\bar{T}}, \hat{x}|_{\bar{T}} \right) $
obtained by measuring the remaining subsystems with $i \in \bar{T}$, where
$\bar{T} = \left \{1, \dots, k \right \} \setminus \left (T_1 \cup T_2  \right
)$. The latter value does not depend on $pw$. Thus, either the protocols aborts
because the error rate exceeded the threshold $\tau$ or the server $P_2$ can
efficiently recover $x|_{I_w}$ independently of $pw$.
It follows that the protocol, either
aborts or succeeds, only depends on the adversary's behavior. The formal proof
is given in Appendix~\ref{proof:no_corrupted}.


\emph{Corrupted client.}
When the client is corrupted, recall that in general the simulator $\adv{S}$
needs to extract the corrupted party's input in order to send it to the
trusted party, and needs to simulate its view so that its output corresponds
to the output received back from the trusted party. Specifically, $\adv{S}$
chooses the CRS from two different distributions corresponding to the
perfectly binding keys and the perfectly hiding keys. The simulated adversary
uses a perfectly binding key included in the CRS (and since it is
quantum-polynomial-time, it
cannot distinguish between the two keys), and $\adv{S}$, upon receipt of
commitments from the adversary, uses the trapdoor information to extract the
committed values and re-commit and output them with a perfectly hiding key.
Furthermore, $\adv{S}$ delays its measurement and only measures its qubits
when needed. In particular, $\adv{S}$ uses its perfectly hiding trapdoor to
equivocate and its perfectly binding trapdoor to extract the password guess of
the adversary. If the guess is correct, then the simulation is perfect and
thus the two states are equal. It thus suffices to argue that the key
$\textsf{sk}$ that the server computes is uniformly random from the view of
the adversary for any fixed $pw' \neq pw$. The upper bound of
indistinguishability follows by privacy amplification (Theorem
\ref{theorem:privacy_amplification}) and private error correction (Theorem
\ref{theorem:private_ecc}). The formal proof is given in Appendix
\ref{proof:corrupted_client}.

\emph{Corrupted server.}
The simulation strategy for a corrupted server is the same as the previous case
and we defer the formal proof to Appendix \ref{proof:corrupted_server}.

\emph{Both parties are corrupted.}
This case is trivial since the simulator $\adv{S}$ just runs the adversary
$\adv{A}$ internally and outputs whatever $\adv{A}$ outputs, hence, the
simulation is perfect.
\end{proof}

\begin{theorem}[main theorem]
\label{theorem:main_theorem_2}
There exists a protocol in the $\ideal{CRS}$-hybrid model that
everlastingly realizes $\ideal{pwKE}$ in the presence of static-corruption
malicious adversaries.
\end{theorem}
\begin{proof}
Follows immediately from Theorem \ref{theorem:main_theorem_1} and
Theorem \ref{theorem:compiler}. In particular, the compiled protocol
$\mathcal{C}^{\varepsilon_{sig}}(\Pi)$ is everlastingly secure against
static-corruption malicious adversaries.
\end{proof}


\section{Conclusion}
We have studied password-authenticated quantum key exchange and proven
its advantage over classical PAKE. The information-theoretic security of
traditional QKD relies on the very strong assumption
regarding the authentication of communication channels. We show here how to
implement this authentication using only passwords. This only decreases the
security from information-theoretic to everlasting, in which the adversary is
supposed to be computationally bounded during the execution of the protocol
but can be unbounded afterwards. This security is still stronger than the
security that can be achieved by classical PAKE protocols, and also relies on
much simpler assumptions and more practical requirements than fully
information-theoretic secure QKD.

%
%
%
%

\bibliographystyle{alpha}
\renewcommand{\doi}[1]{\url{#1}}
\bibliography{abbrev3,crypto,references}

\newcommand{\etalchar}[1]{$^{#1}$}
\begin{thebibliography}{BFGGS12}

\bibitem[ABB{\etalchar{+}}13]{AC:ABBCP13}
Michel Abdalla, Fabrice Benhamouda, Olivier Blazy, C{\'e}line Chevalier, and
  David Pointcheval.
\newblock {SPHF}-friendly non-interactive commitments.
\newblock In Kazue Sako and Palash Sarkar, editors, {\em ASIACRYPT~2013,
  Part~I}, volume 8269 of {\em {LNCS}}, pages 214--234. Springer, Heidelberg,
  December 2013.

\bibitem[ACCP09]{AFRICACRYPT:ACCP09}
Michel Abdalla, Dario Catalano, C{\'e}line Chevalier, and David Pointcheval.
\newblock Password-authenticated group key agreement with adaptive security and
  contributiveness.
\newblock In Bart Preneel, editor, {\em AFRICACRYPT 09}, volume 5580 of {\em
  {LNCS}}, pages 254--271. Springer, Heidelberg, June 2009.

\bibitem[BB84]{BEN84}
C.~H. Bennett and G.~Brassard.
\newblock {Quantum cryptography: Public key distribution and coin tossing}.
\newblock In {\em Proceedings of IEEE International Conference on Computers,
  Systems, and Signal Processing}, page 175, 1984.

\bibitem[BC16]{AC:BlaChe16}
Olivier Blazy and C{\'e}line Chevalier.
\newblock Structure-preserving smooth projective hashing.
\newblock In Jung~Hee Cheon and Tsuyoshi Takagi, editors, {\em ASIACRYPT~2016,
  Part~II}, volume 10032 of {\em {LNCS}}, pages 339--369. Springer, Heidelberg,
  December 2016.

\bibitem[BCL{\etalchar{+}}11]{JC:BCLPR11}
Boaz Barak, Ran Canetti, Yehuda Lindell, Rafael Pass, and Tal Rabin.
\newblock Secure computation without authentication.
\newblock {\em Journal of Cryptology}, 24(4):720--760, October 2011.

\bibitem[BCNP04]{FOCS:BCNP04}
Boaz Barak, Ran Canetti, Jesper~Buus Nielsen, and Rafael Pass.
\newblock Universally composable protocols with relaxed set-up assumptions.
\newblock In {\em 45th FOCS}, pages 186--195. {IEEE} Computer Society Press,
  October 2004.

\bibitem[BCS12]{buhrman2012complete}
Harry Buhrman, Matthias Christandl, and Christian Schaffner.
\newblock Complete insecurity of quantum protocols for classical two-party
  computation.
\newblock {\em Physical review letters}, 109(16):160501, 2012.

\bibitem[BFGGS12]{bouman2012all}
Niek~J Bouman, Serge Fehr, Carlos Gonz{\'a}lez-Guill{\'e}n, and Christian
  Schaffner.
\newblock An all-but-one entropic uncertainty relation, and application to
  password-based identification.
\newblock In {\em Conference on Quantum Computation, Communication, and
  Cryptography}, pages 29--44. Springer, 2012.

\bibitem[BHL{\etalchar{+}}05]{TCC:BHLMO05}
Michael {Ben-Or}, Michal Horodecki, Debbie~W. Leung, Dominic Mayers, and
  Jonathan Oppenheim.
\newblock The universal composable security of quantum key distribution.
\newblock In Joe Kilian, editor, {\em TCC~2005}, volume 3378 of {\em {LNCS}},
  pages 386--406. Springer, Heidelberg, February 2005.

\bibitem[BPR00]{EC:BelPoiRog00}
Mihir Bellare, David Pointcheval, and Phillip Rogaway.
\newblock Authenticated key exchange secure against dictionary attacks.
\newblock In Bart Preneel, editor, {\em EUROCRYPT~2000}, volume 1807 of {\em
  {LNCS}}, pages 139--155. Springer, Heidelberg, May 2000.

\bibitem[BZ13]{C:BonZha13}
Dan Boneh and Mark Zhandry.
\newblock Secure signatures and chosen ciphertext security in a quantum
  computing world.
\newblock In Ran Canetti and Juan~A. Garay, editors, {\em CRYPTO~2013,
  Part~II}, volume 8043 of {\em {LNCS}}, pages 361--379. Springer, Heidelberg,
  August 2013.

\bibitem[Can00]{JC:Canetti00}
Ran Canetti.
\newblock Security and composition of multiparty cryptographic protocols.
\newblock {\em Journal of Cryptology}, 13(1):143--202, January 2000.

\bibitem[Can01]{FOCS:Canetti01}
Ran Canetti.
\newblock Universally composable security: A new paradigm for cryptographic
  protocols.
\newblock In {\em 42nd FOCS}, pages 136--145. {IEEE} Computer Society Press,
  October 2001.

\bibitem[CDVW12]{PKC:CDVW12}
Ran Canetti, Dana {Dachman-Soled}, Vinod Vaikuntanathan, and Hoeteck Wee.
\newblock Efficient password authenticated key exchange via oblivious transfer.
\newblock In Marc Fischlin, Johannes Buchmann, and Mark Manulis, editors, {\em
  PKC~2012}, volume 7293 of {\em {LNCS}}, pages 449--466. Springer, Heidelberg,
  May 2012.

\bibitem[CF01]{C:CanFis01}
Ran Canetti and Marc Fischlin.
\newblock Universally composable commitments.
\newblock In Joe Kilian, editor, {\em CRYPTO~2001}, volume 2139 of {\em
  {LNCS}}, pages 19--40. Springer, Heidelberg, August 2001.

\bibitem[CHK{\etalchar{+}}05]{EC:CHKLM05}
Ran Canetti, Shai Halevi, Jonathan Katz, Yehuda Lindell, and Philip~D.
  {MacKenzie}.
\newblock Universally composable password-based key exchange.
\newblock In Ronald Cramer, editor, {\em EUROCRYPT~2005}, volume 3494 of {\em
  {LNCS}}, pages 404--421. Springer, Heidelberg, May 2005.

\bibitem[CLOS02]{STOC:CLOS02}
Ran Canetti, Yehuda Lindell, Rafail Ostrovsky, and Amit Sahai.
\newblock Universally composable two-party and multi-party secure computation.
\newblock In {\em 34th ACM STOC}, pages 494--503. {ACM} Press, May 2002.

\bibitem[DFL{\etalchar{+}}09]{C:DFLSS09}
Ivan Damg{\aa}rd, Serge Fehr, Carolin Lunemann, Louis Salvail, and Christian
  Schaffner.
\newblock Improving the security of quantum protocols via commit-and-open.
\newblock In Shai Halevi, editor, {\em CRYPTO~2009}, volume 5677 of {\em
  {LNCS}}, pages 408--427. Springer, Heidelberg, August 2009.

\bibitem[DFSS05]{FOCS:DFSS05}
Ivan Damg{\aa}rd, Serge Fehr, Louis Salvail, and Christian Schaffner.
\newblock Cryptography in the bounded quantum-storage model.
\newblock In {\em 46th FOCS}, pages 449--458. {IEEE} Computer Society Press,
  October 2005.

\bibitem[DFSS07]{C:DFSS07b}
Ivan Damg{\aa}rd, Serge Fehr, Louis Salvail, and Christian Schaffner.
\newblock Secure identification and {QKD} in the bounded-quantum-storage model.
\newblock In Alfred Menezes, editor, {\em CRYPTO~2007}, volume 4622 of {\em
  {LNCS}}, pages 342--359. Springer, Heidelberg, August 2007.

\bibitem[DH76]{DifHel76}
Whitfield Diffie and Martin~E. Hellman.
\newblock New directions in cryptography.
\newblock {\em {IEEE} Transactions on Information Theory}, 22(6):644--654,
  1976.

\bibitem[DS05]{STOC:DodSmi05}
Yevgeniy Dodis and Adam Smith.
\newblock Correcting errors without leaking partial information.
\newblock In Harold~N. Gabow and Ronald Fagin, editors, {\em 37th ACM STOC},
  pages 654--663. {ACM} Press, May 2005.

\bibitem[FKS{\etalchar{+}}13]{TCC:FKSZZ13}
Serge Fehr, Jonathan Katz, Fang Song, Hong-Sheng Zhou, and Vassilis Zikas.
\newblock Feasibility and completeness of cryptographic tasks in the quantum
  world.
\newblock In Amit Sahai, editor, {\em TCC~2013}, volume 7785 of {\em {LNCS}},
  pages 281--296. Springer, Heidelberg, March 2013.

\bibitem[FS08]{TCC:FehSch08}
Serge Fehr and Christian Schaffner.
\newblock Randomness extraction via {$\delta$}-biased masking in the presence
  of a quantum attacker.
\newblock In Ran Canetti, editor, {\em TCC~2008}, volume 4948 of {\em {LNCS}},
  pages 465--481. Springer, Heidelberg, March 2008.

\bibitem[FS09]{TCC:FehSch09}
Serge Fehr and Christian Schaffner.
\newblock Composing quantum protocols in a classical environment.
\newblock In Omer Reingold, editor, {\em TCC~2009}, volume 5444 of {\em
  {LNCS}}, pages 350--367. Springer, Heidelberg, March 2009.

\bibitem[FS17]{EC:FehSal17}
Serge Fehr and Louis Salvail.
\newblock Quantum authentication and encryption with key recycling - or: How to
  re-use a one-time pad even if {P}={NP} - safely \& feasibly.
\newblock In Jean{-}S{\'{e}}bastien Coron and Jesper~Buus Nielsen, editors,
  {\em EUROCRYPT~2017, Part~III}, volume 10212 of {\em {LNCS}}, pages 311--338.
  Springer, Heidelberg, April~/~May 2017.

\bibitem[GL03]{EC:GenLin03}
Rosario Gennaro and Yehuda Lindell.
\newblock A framework for password-based authenticated key exchange.
\newblock In Eli Biham, editor, {\em EUROCRYPT~2003}, volume 2656 of {\em
  {LNCS}}, pages 524--543. Springer, Heidelberg, May 2003.
\newblock \url{http://eprint.iacr.org/2003/032.ps.gz}.

\bibitem[GMR88]{GolMicRiv88}
Shafi Goldwasser, Silvio Micali, and Ronald~L. Rivest.
\newblock A digital signature scheme secure against adaptive chosen-message
  attacks.
\newblock {\em {SIAM} Journal on Computing}, 17(2):281--308, April 1988.

\bibitem[Gol01]{Goldreich01}
Oded Goldreich.
\newblock {\em Foundations of Cryptography: Basic Tools}, volume~1.
\newblock Cambridge University Press, Cambridge, UK, 2001.

\bibitem[GOS12]{Groth:2012:NTN:2220357.2220358}
Jens Groth, Rafail Ostrovsky, and Amit Sahai.
\newblock New techniques for noninteractive zero-knowledge.
\newblock {\em J. ACM}, 59(3):11:1--11:35, June 2012.

\bibitem[Hoe63]{hoeffding1963probability}
Wassily Hoeffding.
\newblock Probability inequalities for sums of bounded random variables.
\newblock {\em Journal of the American statistical association},
  58(301):13--30, 1963.

\bibitem[JR15]{AC:JutRoy15}
Charanjit~S. Jutla and Arnab Roy.
\newblock Dual-system simulation-soundness with applications to {UC}-{PAKE} and
  more.
\newblock In Tetsu Iwata and Jung~Hee Cheon, editors, {\em ASIACRYPT~2015,
  Part~I}, volume 9452 of {\em {LNCS}}, pages 630--655. Springer, Heidelberg,
  November~/~December 2015.

\bibitem[JR18]{TCC:JutRoy18}
Charanjit~S. Jutla and Arnab Roy.
\newblock Smooth {NIZK} arguments.
\newblock In Amos Beimel and Stefan Dziembowski, editors, {\em TCC~2018,
  Part~I}, volume 11239 of {\em {LNCS}}, pages 235--262. Springer, Heidelberg,
  November 2018.

\bibitem[Kil88]{STOC:Kilian88}
Joe Kilian.
\newblock Founding cryptography on oblivious transfer.
\newblock In {\em 20th ACM STOC}, pages 20--31. {ACM} Press, May 1988.

\bibitem[KLR06]{STOC:KusLinRab06}
Eyal Kushilevitz, Yehuda Lindell, and Tal Rabin.
\newblock Information-theoretically secure protocols and security under
  composition.
\newblock In Jon~M. Kleinberg, editor, {\em 38th ACM STOC}, pages 109--118.
  {ACM} Press, May 2006.

\bibitem[KMQ11]{TCC:KraMul11}
Daniel Kraschewski and J{\"o}rn M{\"u}ller-Quade.
\newblock Completeness theorems with constructive proofs for finite
  deterministic 2-party functions.
\newblock In Yuval Ishai, editor, {\em TCC~2011}, volume 6597 of {\em {LNCS}},
  pages 364--381. Springer, Heidelberg, March 2011.

\bibitem[KOY01]{EC:KatOstYun01}
Jonathan Katz, Rafail Ostrovsky, and Moti Yung.
\newblock Efficient password-authenticated key exchange using human-memorable
  passwords.
\newblock In Birgit Pfitzmann, editor, {\em EUROCRYPT~2001}, volume 2045 of
  {\em {LNCS}}, pages 475--494. Springer, Heidelberg, May 2001.

\bibitem[KV11]{TCC:KatVai11}
Jonathan Katz and Vinod Vaikuntanathan.
\newblock Round-optimal password-based authenticated key exchange.
\newblock In Yuval Ishai, editor, {\em TCC~2011}, volume 6597 of {\em {LNCS}},
  pages 293--310. Springer, Heidelberg, March 2011.

\bibitem[LC97]{lo1997quantum}
Hoi-Kwong Lo and Hoi~Fung Chau.
\newblock Is quantum bit commitment really possible?
\newblock {\em Physical Review Letters}, 78(17):3410, 1997.

\bibitem[Lo97]{lo1997insecurity}
Hoi-Kwong Lo.
\newblock Insecurity of quantum secure computations.
\newblock {\em Physical Review A}, 56(2):1154, 1997.

\bibitem[May97]{mayers1997unconditionally}
Dominic Mayers.
\newblock Unconditionally secure quantum bit commitment is impossible.
\newblock {\em Physical review letters}, 78(17):3414, 1997.

\bibitem[MPR10]{C:MajPraRos10}
Hemanta~K. Maji, Manoj Prabhakaran, and Mike Rosulek.
\newblock A zero-one law for cryptographic complexity with respect to
  computational {UC} security.
\newblock In Tal Rabin, editor, {\em CRYPTO~2010}, volume 6223 of {\em {LNCS}},
  pages 595--612. Springer, Heidelberg, August 2010.

\bibitem[MQU07]{TCC:QuaUnr07}
J{\"o}rn M{\"u}ller-Quade and Dominique Unruh.
\newblock Long-term security and universal composability.
\newblock In Salil~P. Vadhan, editor, {\em TCC~2007}, volume 4392 of {\em
  {LNCS}}, pages 41--60. Springer, Heidelberg, February 2007.

\bibitem[MR11]{Maurer11abstractcryptography}
Ueli Maurer and Renato Renner.
\newblock Abstract cryptography.
\newblock In {\em In Innovations in Computer Science}. Tsinghua University
  Press, 2011.

\bibitem[MSU13]{PQCRYPTO:MosSteUst13}
Michele Mosca, Douglas Stebila, and Berkant Ustaoglu.
\newblock Quantum key distribution in the classical authenticated key exchange
  framework.
\newblock In Philippe Gaborit, editor, {\em Post-Quantum Cryptography - 5th
  International Workshop, PQCrypto 2013}, pages 136--154. Springer, Heidelberg,
  June 2013.

\bibitem[NC11]{nielsen2002quantum}
Michael~A. Nielsen and Isaac~L. Chuang.
\newblock {\em Quantum Computation and Quantum Information: 10th Anniversary
  Edition}.
\newblock Cambridge University Press, New York, NY, USA, 10th edition, 2011.

\bibitem[Por17]{EC:Portmann17}
Christopher Portmann.
\newblock Quantum authentication with key recycling.
\newblock In Jean{-}S{\'{e}}bastien Coron and Jesper~Buus Nielsen, editors,
  {\em EUROCRYPT~2017, Part~III}, volume 10212 of {\em {LNCS}}, pages 339--368.
  Springer, Heidelberg, April~/~May 2017.

\bibitem[Ren05]{quant-ph/0512258}
Renato Renner.
\newblock {\em Security of Quantum Key Distribution}.
\newblock PhD thesis, 2005.

\bibitem[RK05]{TCC:RenKon05}
Renato Renner and Robert K{\"o}nig.
\newblock Universally composable privacy amplification against quantum
  adversaries.
\newblock In Joe Kilian, editor, {\em TCC~2005}, volume 3378 of {\em {LNCS}},
  pages 407--425. Springer, Heidelberg, February 2005.

\bibitem[SSS09]{AC:SalSchSot09}
Louis Salvail, Christian Schaffner, and Miroslava Sot{\'a}kov{\'a}.
\newblock On the power of two-party quantum cryptography.
\newblock In Mitsuru Matsui, editor, {\em ASIACRYPT~2009}, volume 5912 of {\em
  {LNCS}}, pages 70--87. Springer, Heidelberg, December 2009.

\bibitem[TL17]{Tomamichel2017largelyself}
Marco Tomamichel and Anthony Leverrier.
\newblock A largely self-contained and complete security proof for quantum key
  distribution.
\newblock {\em {Quantum}}, 1:14, July 2017.

\bibitem[TR11]{tomamichel2011uncertainty}
Marco Tomamichel and Renato Renner.
\newblock Uncertainty relation for smooth entropies.
\newblock {\em Physical review letters}, 106(11):110506, 2011.

\bibitem[Unr10]{EC:Unruh10}
Dominique Unruh.
\newblock Universally composable quantum multi-party computation.
\newblock In Henri Gilbert, editor, {\em EUROCRYPT~2010}, volume 6110 of {\em
  {LNCS}}, pages 486--505. Springer, Heidelberg, May~/~June 2010.

\bibitem[Unr13]{C:Unruh13}
Dominique Unruh.
\newblock Everlasting multi-party computation.
\newblock In Ran Canetti and Juan~A. Garay, editors, {\em CRYPTO~2013,
  Part~II}, volume 8043 of {\em {LNCS}}, pages 380--397. Springer, Heidelberg,
  August 2013.

\bibitem[WTHR11]{winkler2011impossibility}
Severin Winkler, Marco Tomamichel, Stefan Hengl, and Renato Renner.
\newblock Impossibility of growing quantum bit commitments.
\newblock {\em Physical review letters}, 107(9):090502, 2011.

\end{thebibliography}

\appendix

\section{Preliminaries}
\subsection{Security Models}\label{app:models}

We provide a brief overview of security models for
 \emph{multi-party computation} (MPC), in which $n$
players interact in order to compute securely a given function of their
inputs. Formally, consider $n$ players $P_i$, each owning an input $x_i$,
and a classical $n$-input function $f$.
The goal is to compute $\left( y_1, \dots, y_n
\right) = f \left( x_1, \dots, x_n \right) $ such that each player $P_i$
learns $y_i$, and cheating players cannot change the outcome of the
computation (apart by choosing a different input) and do not learn more
about the input (and possibly the output) of honest players than what can be
derived from their own input and their output of the function evaluation.

\subsubsection{The Simulation-based Paradigm.} The first step towards the
solution for this security definition is the \emph{simulation} paradigm.
Instead of introducing different notion for each security property,
we consider for each protocol, the ``ideal behavior'' it should have.
Intuitively, we introduce the notion of ``ideal world'' where
there is a trusted party who collects the inputs from all players, computes
the output and distributes the output to the players. A real protocol is
compared to an ideal protocol, and  the real protocol is said to be \emph{at
least as secure as} the ideal protocol if the real protocol and the ideal
protocol have an indistinguishable input-output behavior. The level of
security reached thus also depends on the specification of the ideal protocol.


\subsubsection{Universal Composability.}
However, as being pointed out in the literature, the simulation-based paradigm
does not play well with \emph{composition} and in fact, it
only achieves \emph{Sequential Composition}, \ie{}, a protocol that is secure
under sequential composition maintains its security when run multiple times,
as long as the executions are run sequentially (meaning that each execution
concludes before the next execution begins). In the case of \emph{Concurrent
Composition} in which many instances of the same protocol with correlated
inputs are run concurrently, some problems may occur. For example, the
messages from one protocol could be fed into another, or a message from one
sub-protocol of a larger application is fed into another sub-protocol and the
overall application becomes insecure.
In order to solve this inherent problem, the so-called UC (for \emph{Universal
Composability}) framework was introduced. We give a high-level overview of the
model below and refer
the reader to \cite{FOCS:Canetti01} for more details on the classical version
and \cite{EC:Unruh10} for the quantum version.


\paragraph{Ideal World and Real World.}
We define in the ideal world an entity that one can never corrupt,
called the \emph{ideal functionality} and usually denoted as $\ideal{}$. The
players privately send their inputs to this entity, and receive their
corresponding output the same way. There is no communication between the
different players. $\ideal{}$ is assumed to behave in a perfectly correct
way, without revealing information other than required, and without being
possibly corrupted by an adversary. Once $\ideal{}$ is defined, the goal of
a protocol $\pi$, executed in a real world in the presence of an adversary,
is then to create a situation equivalent to that obtained with $\ideal{}$.

\paragraph{Protocol, Adversary, and Environment.} Apart from the protocol
participants which are specified by the protocol, there are two more machines
taking part in the protocol execution. The \emph{adversary} $\adv{A}$ (or
$\adv{S}$ in the ideal model) is the machine coordinating all corrupted
participants analogous to the simulation-based model. The \emph{environment
machine} $\adv{Z}$, playing the role of the \emph{distinguisher}, models
``everything that is outside the protocol being executed''. It chooses the
inputs, sees the outputs, and may communicate with the adversary at any time.
The adversary has access to the communication between players, but not to the
inputs and outputs of the honest players (it completely controls the dishonest
or corrupted players). On the contrary, the environment has access to the
inputs and outputs of all players, but not to their communication, nor to the
inputs and outputs of the subroutines they can invoke.

A protocol $\pi$ \emph{securely realizes} a functionality $\ideal{}$
if for every real-world adversary $\adv{A}$ there exists an ideal-world
adversary $\adv{S}$, called the simulator, such that no environment can
distinguish whether it is witnessing the real-world execution with adversary
$\adv{A}$ or the ideal-world execution with simulator $\adv{S}$, with a
non-negligible advantage.
Depending on the assumed computing power of the adversary and the
environment we distinguish between \emph{computational} security, where they
are all considered to be polynomially bounded machines, and \emph{statistical}
security, where they are assumed to be computationally unbounded. Furthermore,
in \cite{C:Unruh13}, Unruh introduces the notion of
\emph{everlasting} security, where the adversary is considered to be
a polynomial-time machine but the environment is assumed to have unbounded
computational power.

In addition, the notion of ``hybrid models'' is also introduced to model
the concept of set-up assumptions. A protocol $\pi$ is said to be realized
``in the $\ideal{}$-hybrid model'' if $\pi$ can invoke the ideal
functionality $\ideal{}$ as a subroutine multiple times. We note that the
environment can never interact directly with $\ideal{}$, and thus,
$\ideal{}$ is usually never invoked at all in the ideal world, and the
implementation of $\ideal{}$ is simulated solely by the ideal adversary
$\adv{S}$.
The model with no trusted set-up is called \emph{plain}.

\paragraph{Ideal Functionalities.}
\label{app:functionalities}

We denote $\ideal{CRS}$ the common reference string functionality,
$\ideal{OT}$ the oblivious transfer functionality, $\ideal{COM}$ the bit
commitment functionality, and $\ideal{pwKE}$ the password-based key-exchange functionality. The definitions of these functionalities are given
as follows.

The common reference string (CRS) model is modeled by the functionality
$\mathcal{F}^{\mathcal{D}}_{\textsf{CRS}}$, which was presented
in \cite{FOCS:BCNP04}. At each call of
$\mathcal{F}^{\mathcal{D}}_{\textsf{CRS}}$, it sends back the same reference
string, chosen by itself, following a known public distribution $\mathcal{D}$.
We recall it here in Figure \ref{fig:ideal_crs}.

\begin{figure}[!htb]
\begin{tcolorbox}[colback=white,size=small,sharp corners,colframe=black]
{\centering
  \textbf{Functionality} $\mathcal{F}^{\mathcal{D}}_{\textsf{CRS}}$ \\
} The functionality $\mathcal{F}^{\mathcal{D}}_{\textsf{CRS}}$ is
parameterized by a distribution $\mathcal{D}$. It interacts with a set of
players and an adversary in the following way:
\begin{itemize}
\item Choose a value $r \sample \mathcal{D}$.
\item Upon receiving a value $\left( \textsf{CRS}, \textsf{sid} \right) $ from
a player, send $\left( \textsf{CRS}, \textsf{sid}, r \right)$ to this player.
\end{itemize}
\end{tcolorbox}
\caption{The functionality $\mathcal{F}^{\mathcal{D}}_{\textsf{CRS}}$.}
\label{fig:ideal_crs}
\end{figure}

Next, we present the ideal functionality for bit commitment protocols in Figure
\ref{fig:ideal_com}, which is due to \cite{C:CanFis01}.
\begin{figure}[!htb]
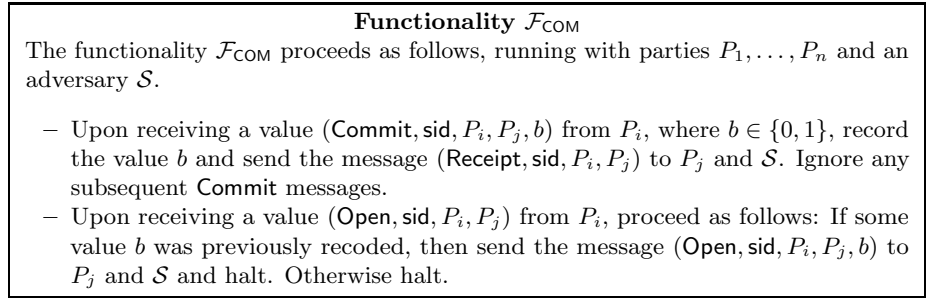

\begin{tcolorbox}[colback=white,size=small,sharp corners,colframe=black]
{\centering
  \textbf{Functionality} $\ideal{COM}$ \\
}

The functionality $\ideal{COM}$ proceeds as follows, running with parties
$P_1, \dots, P_n$ and an adversary $\adv{S}$.
\begin{itemize}
\item Upon receiving a value $\left( \textsf{Commit}, \textsf{sid}, P_i, P_j,
b \right) $ from $P_i$, where $b \in \bin$, record the value $b$ and send the
message $\left( \textsf{Receipt}, \textsf{sid}, P_i, P_j \right) $ to $P_j$
and $\adv{S}$. Ignore any subsequent $\textsf{Commit}$ messages.
\item Upon receiving a value $\left( \textsf{Open}, \textsf{sid}, P_i, P_j
\right) $ from $P_i$, proceed as follows: If some value $b$ was previously
recoded, then send the message $\left( \textsf{Open}, \textsf{sid}, P_i, P_j,
b \right) $ to $P_j$ and $\adv{S}$ and halt. Otherwise halt.
\end{itemize}
\end{tcolorbox}
\caption{The Ideal Commitment functionality for a single commitment.}
\label{fig:ideal_com}
\mbox{}\\[3mm]
\end{figure}

Oblivious Transfer (OT) is a very powerful tool and is sufficient to realize
any secure computation functionality \cite{STOC:Kilian88}. Informally, OT is a
two-party functionality, involving a sender $S$ with input $x_0, x_1$ and a
receiver $R$ with an input $\sigma \in \bin$. The receiver $R$ learns
$x_\sigma$ (and nothing else), and the sender learns nothing at all. These
requirements are captured by the specification of the OT functionality
$\ideal{OT}$ from \cite{STOC:CLOS02}, given in Figure \ref{fig:ideal_ot}.
\begin{figure}[!htbp]
\begin{tcolorbox}[colback=white,size=small,sharp corners,colframe=black]
{\centering
  \textbf{Functionality} $\ideal{OT}$ \\
}

The functionality $\ideal{OT}$ interacts with a sender $S$, a receiver $R$
and an adversary $\adv{S}$.
\begin{itemize}
\item Upon receiving a message $\left( \textsf{sender}, \textsf{sid}, x_0, x_1
\right) $ from $S$, where each $x_i \in \bin^\ell$, store $\left( x_0, x_1
\right) $. (The lengths of the strings $\ell$ is fixed and known to all
parties).
\item Upon receiving a message $\left( \textsf{receiver}, \textsf{sid}, \sigma
\right) $ from $R$, check if a $\textsf{sender}$ message was previously sent.
If yes, send $\left( \textsf{sid}, x_\sigma \right) $ to $R$ and
$\textsf{sid}$ to the adversary $\adv{S}$ and halt. If not, send nothing to
$R$ (but continue running).
\end{itemize}
\end{tcolorbox}
\caption{The oblivious transfer functionality $\ideal{OT}$.}
\label{fig:ideal_ot}
\end{figure}

Our definition of the password-based key-exchange functionality $\ideal{pwKE}$
(Figure~\ref{fig:f_pwke}) is identical to the description in \cite{EC:CHKLM05}.
A natural property of PAKE is that due to the low entropy of passwords, PAKE
protocols are subject to dictionary attacks. The adversary can
break the security of the scheme by trying all values for the password in the
small set of the possible values (\ie{}, the dictionary). Unfortunately, these
attacks can be quite damaging since the attacker has a non-negligible
probability of succeeding. To address this problem, one should invalidate or
block the use of a password whenever a certain number of failed attempts
occurs. However, this is only effective in the case of \emph{online}
dictionary attacks in which the adversary must be present and interact with
the system in order to be able to verify whether its guess is correct. Thus,
the goal of PAKE protocol is to restrict the adversary to
\emph{online} dictionary attacks only. In other words, \emph{off-line}
dictionary attacks, in which the adversary verifies if a password guess is
correct without interacting with the system, should not be possible in a PAKE
protocol.

We refer the reader to \cite{EC:CHKLM05} for motivating discussion regarding
the particular choices made in this formulation of the functionality.
In particular, this formulation captures PAKE protocols with \emph{implicit}
authentication (the version with \emph{explicit} authentication being
described on \Cref{fig:f_epwKE}).

\begin{figure}[!htbp]
\begin{tcolorbox}[colback=white,size=small,sharp corners,colframe=black]
{\centering
	\textbf{The functionality} $\ideal{pwKE}$ \\
}

The functionality $\ideal{pwKE}$ is parameterized by a security parameter $\lambda$.
It interacts with an adversary $\adv{S}$ and a set of parties via the
following queries:
\begin{itemize}[leftmargin=*,label={}]
\item \textbf{Upon receiving a query} $(\textsf{NewSession}, sid, P_i, P_j, pw, \textnormal{role})$ \textbf{from party $P_i$:}
\begin{itemize}[label={}]
\item Send $\left( \textsf{NewSession}, sid, P_i, P_j, \textnormal{role}
\right) $ to $\adv{S}$. In addition, if this is the first
$\textsf{NewSession}$ query, or if this is the second $\textsf{NewSession}$
query and there is a record $\left( P_j, P_i, pw' \right) $, then record
$\left( P_i, P_j, pw \right) $ and mark this record $\textsf{fresh}$.
\end{itemize}
\item \textbf{Upon receiving a query} $(\textsf{TestPwd}, sid, P_i, pw')$ \textbf{from the adversary} $\mathcal{S}$ \textbf{:}
\begin{itemize}[label={}]
\item If there is a record of the form $\left( P_i, P_j, pw \right) $ which is
$\textsf{fresh}$, then do: If $pw = pw'$, mark the record
$\textsf{compromised}$ and reply to $\adv{S}$ with ``correct guess''. If $pw
\neq pw'$, mark the record $\textsf{interrupted}$ and reply with ``wrong
guess''.
\end{itemize}
\item \textbf{Upon receiving a query} $(\textsf{NewKey}, sid, P_i, sk)$ \textbf{from} $\adv{S}$\textbf{, where } $\size{sk} = \lambda$ \textbf{:}
\begin{itemize}[label={}]
\item If there is a record of the form $\left( P_i, P_j, pw \right) $, and this
is the first $\textsf{NewKey}$ for $P_i$, then:
\begin{itemize}[leftmargin=*,label={$\bullet$}]
\item If this record is $\textsf{compromised}$, or either $P_i$ or $P_j$ is
corrupted, then output $\left( sid, sk \right) $ to player $P_i$.
\item If this record is $\textsf{fresh}$, and there is a record $(P_j, P_i,
pw')$ with $pw = pw'$, and a key $sk'$ was already sent to $P_j$, and
$\left(P_j, P_i, pw \right) $ was $\textsf{fresh}$ at the time, then output
$\left(sid, sk' \right)$ to $P_i$.
\item In any other case, pick a new random key $sk'$ of length $\lambda$
and send $\left( sid, sk' \right) $ to $P_i$.
\end{itemize}
\item Either way, mark the record $\left( P_i, P_j, pw \right) $ as
$\textsf{completed}$.
\end{itemize}
\end{itemize}
\end{tcolorbox}
\caption{The password-based key-exchange functionality $\ideal{pwKE}$.}
\label{fig:f_pwke}
\end{figure}

\subsection{Cryptographic Primitives}\label{app:primitives}

{

\renewcommand{\S}{\mathcal{S}}

A digital signature scheme~\cite{DifHel76,GolMicRiv88} allows a signer
to produce a verifiable proof that he indeed produced a message. We here
consider signatures that are resistant to a quantum chosen message
attack, as defined by Boneh and Zhandry in~\cite{C:BonZha13}. We recall
the definition and security notion below.

\begin{definition}
A signature scheme is a tuple of efficient classical
algorithms $(\KeyGen, \Sign, \Verify)$, where

        \begin{itemize}

                \item $\KeyGen(\seck)$, where $\seck$ is the security
parameter, outputs a pair $(\sk,\vk)$, where $\sk$ is the (secret)
signing key, and $\vk$ is the (public) verification key;

                \item $\Sign_{\sk}(M;\mu)$, outputs a signature
$\sigma(M)$, on a message $M$, under the signing key $\sk$ and
randomness~$\mu$;

                \item $\Verif_{\vk}(M,\sigma)$ checks the validity of the
signature $\sigma$ with respect to the message $M$ and the verification
key $\vk$. And so outputs a bit.

        \end{itemize}
\end{definition}

The properties of a digital signature scheme can be defined as follows:

\begin{itemize}
        \item \emph{Correctness}: For every pair $(\vk,\sk)$ generated
by $\KeyGen$, for every message $M$ and for every random $\mu$, we have
$\Verif_{\vk}(M,\Sign_{\sk}(M;\mu))=1$.

	\item \emph{Existential unforgeability under adaptive quantum
chosen-message attack}: a signature scheme $(\KeyGen, \Sign, \Verify)$ is
strongly existentially unforgeable under a quantum chosen-message attack
(EUF-qCMA secure) if, for any efficient quantum algorithm~$A$ and
any polynomial~$q$, $A$'s probability of success in the following
game is negligible in~$\secpar$:

\begin{itemize}

\item \textbf{Key Generation:} The challenger runs $(\sk,\pk) \leftarrow
\KeyGen(\secpar)$, and gives $\vk$ to~$A$.

\item \textbf{Signing queries:} The adversary makes a polynomial $q$
chosen message queries. For each query, the challenger chooses
randomness~$r$, and responds by signing each message in the query
using~$r$ as randomness:
\[
	\sum_{m,t} \psi_{m,t} \left| m,t \right\rangle
	\rightarrow
	\sum_{m,t} \psi_{m,t} \left| m,t \oplus \Sign_{\sk}(M;r)\right\rangle
\]

\item \textbf{Forgeries:} The adversary is required to produce $q+1$
message/signature pairs. The challenger then checks that all signatures
are valid, and that all message/signature pairs are distinct. If so, the
challenger reports that the adversary wins.

\end{itemize}

\end{itemize}

}

\subsubsection{Dual-mode Commitment.}\label{sec:dual-mode} We give here an 
informal 
security
definitions for commitment schemes, and refer the reader to
\cite{Goldreich01} for a formal definition.
A commitment scheme $\com$ is defined by 3 algorithms:
\begin{itemize}
\item $\com.\KeyGen{}\left( 1^\lambda \right) $, where $\lambda$ is the
security parameter, generates the global parameters $\textsf{param}$ of the
scheme (which includes the commitment key), implicitly given as input to the
other algorithms;
\item $\com.\Commit\left( m;r \right)$ produces a commitment $c$ on the
input message $m$ from a message space $\mathcal{M}$, using the random coins
$r$ from a randomizer space $\mathcal{R}$, and also outputs the opening
information $w$;
\item $\com.\Verify\left( c, m; w \right)$ verifies the commitment $c$ of the
message $m$
using the opening information $w$; it outputs the message $m$, or $\perp$ if
the opening check fails.
\end{itemize}

To be useful in practice, a commitment scheme should satisfy two basic
security properties. The first one is \emph{hiding}, which informally
guarantees that no information about $m$ is leaked through the commitment $c$.
The second one is \emph{binding}, which guarantees that the committer cannot
generate a commitment $c$ that can be successfully opened to two different
messages. A commitment can be either \emph{perfect hiding} (in which case it
is perfectly secure from the committer's point of view) or \emph{perfect
binding} (in which case it is perfectly secure from the receiver's point of
view). Interestingly, it is proven that information-theoretically secure
commitment protocols (which are both perfect hiding and perfect binding)
cannot exist classically, nor even if we allow to use
quantum mechanics \cite{lo1997quantum,mayers1997unconditionally}.

Our construction uses a non-interactive commitment scheme with some special
properties. This scheme, with a quantum-safe construction based on lattice
assumptions, is given in \cite{C:DFLSS09}.
First, we want a commitment scheme that has two different flavors
of keys, where the corresponding commitment key is generated by one of two
possible key-generation algorithms: $\KeyGen{}_\fsub{H}$ or
$\KeyGen{}_\fsub{B}$. For a key generated by $\KeyGen{}_\fsub{H}$,
the commitment scheme is perfectly hiding, in which case the commitment
reveals no information about the message. Alternatively, the
commitment key generated by $\KeyGen{}_\fsub{B}$ can be perfectly
binding, in which case a valid commitment uniquely defines one possible
message. Both key generation algorithms are probabilistic polynomial time.
They output a commitment key and also some trapdoor information such that we
can either open a commitment to any message (if the commitment key is perfectly
hiding, \ie{}, generated by $\KeyGen{}_\fsub{H}$), or efficiently extract the
committed value (if the commitment key is
perfectly binding, \ie{}, generated by $\KeyGen{}_\fsub{B}$). Furthermore, we
require that keys generated by $\KeyGen{}_\fsub{H}$ and $\KeyGen{}_\fsub{B}$
are computationally indistinguishable, even against quantum adversaries.

The formal definition of dual-mode commitment
scheme~\cite{Groth:2012:NTN:2220357.2220358} is given as follows. For
simplicity and efficiency, we consider the common reference string
model, and we assume the commitment key to be contained in the CRS.



\begin{definition}
$\com = (\KeyGen_\fsub{H}, \KeyGen_\fsub{B}, \Commit, \Verify,
\fname{Open}, \fname{Ext})$ is a dual-mode commitment scheme if it is
a standard commitment scheme with the two additional algorithms:

\begin{itemize} \item $\fname{Open}_{tk}\left( m_1, r_1, m_2 \right)$ on
the input messages $m_1$ and $m_2$ from a message space $\mathcal{M}$,
and a random coin $r_1$ from a randomizer space $\mathcal{R}$, outputs a
random coin $r_2$ such that $\Commit\left( m_1; r_1 \right) =
\Commit\left( m_2; r_2 \right) $. Also outputs the opening information
$w_2 = r_2$ for the second commitment. This algorithm uses a trapdoor $tk$
when the key is hiding;

\item $\fname{Ext}_{xk}\left( c \right)$ on input a commitment $m$,
outputs the message~$m$. This algorithm uses a trapdoor $xk$
when the key is binding.

\end{itemize}

\noindent and also satisfies the following properties for all
non-uniform quantum polynomial time adversaries $\adv{A}$.

\emph{\textbf{Key indistinguishability:}}
\begin{align*}
\pr{\left( ck, xk \right) \leftarrow \KeyGen_\fsub{B}(1^k) \colon \mathcal{A}\left( ck \right)  = 1} \distinguish{\approx}{c} \pr{\left( ck, tk \right)  \leftarrow \KeyGen_\fsub{H}(1^k) \colon \mathcal{A}\left( ck \right)  = 1}
\end{align*}

\emph{\textbf{Perfect binding:}}
\begin{align*}
\textnormal{Pr}\bigl[&(ck, xk) \leftarrow \KeyGen_\fsub{B}(1^k) \colon\exists (m_1, r_1), (m_2, r_2) \in \mathcal{M} \times \mathcal{R}\\
&\textnormal{ such that } m_1 \neq m_2 \textnormal{ and } \Commit\left( m_1; r_1 \right)  = \Commit\left( m_2; r_2 \right) \bigr] = 0
\end{align*}

\emph{\textbf{Perfect extractability:}}
\begin{align*}
\pr{(ck, xk) \leftarrow \KeyGen_\fsub{B}(1^k) \colon \forall (m, r) \in \mathcal{M} \times \mathcal{R}\colon \fname{Ext}_{xk}\left( \Commit\left( m; r \right)  \right)  = m} = 1
\end{align*}

\emph{\textbf{Perfect hiding:}}
\begin{align*}
\pr{(ck_1, tk_1) \leftarrow \KeyGen_\fsub{H}(1^k) \colon \mathcal{A}\left( ck_1, \Commit\left( m_1; * \right) \right)  = 1} \\
= \quad \pr{(ck_2, tk_2) \leftarrow \KeyGen_\fsub{H}(1^k) \colon \mathcal{A}\left( ck_2, \Commit\left( m_2; * \right) \right) = 1}
\end{align*}

\emph{\textbf{Perfect trapdoor opening:}}
\begin{align*}
\textnormal{Pr}\bigl[&(ck, tk) \leftarrow \KeyGen_\fsub{H}(1^k) ; (m_1, m_2) \leftarrow \mathcal{A}\left( ck \right) ; r_1 \leftarrow \mathcal{R}; r2 \leftarrow \fname{Open}_{tk}\left( m_1, r_1, m_2 \right) : \\
&\Commit\left( m_1; r_1 \right)  = \Commit\left( m_2; r_2 \right)  \textnormal{ if } m_1, m_2 \in \mathcal{M}\bigr] = 1
\end{align*}

\emph{\textbf{Perfect trapdoor opening indistinguishability:}}
\begin{align*}
\textnormal{Pr}\bigl[&(ck, tk) \leftarrow \KeyGen_\fsub{H}(1^k) ; (m_1, m_2) \leftarrow \mathcal{A}\left( ck \right) ; r_1 \leftarrow \mathcal{R}; r_2 \leftarrow \fname{Open}_{tk}\left( m_1, r_1, m_2 \right) \colon \\
& (m_1, m_2) \in \mathcal{M} \textnormal{ and } \mathcal{A}(r_1) = 1 \bigr] \\
= \quad \textnormal{Pr}\bigl[&(ck, tk) \leftarrow \KeyGen_\fsub{H}(1^k) ; (m_1, m_2) \leftarrow \mathcal{A}(ck); r_2 \leftarrow \mathcal{R}\colon \\
& (m_1, m_2) \in \mathcal{M} \textnormal{ and } \mathcal{A}(r_2) = 1 \bigr]
\end{align*}
\end{definition}

\section{Proofs of the Impossibility Results for PAKE: Reduction from
EQUALITY to PAKE}
\label{sec:impossibility-proofs}
\label{app:impossibility-proofs}

\subsection{Reduction from EQUALITY to PAKE}


We now prove the impossibility results stated in
Section~\ref{sec:impossibility}, by reducing the problem of constructing
a scheme for the PAKE functionality to the problem of constructing a
scheme for an equality-testing functionality $\ideal{EQ}$.


We consider an explicit mutual authentication PAKE functionality
$\ideal{e-pwKE}$ whose description is given in \Cref{fig:f_epwKE}.
The description of the functionality is a modified version
of the description in \cite{PKC:CDVW12,AFRICACRYPT:ACCP09}.

\afterpage{%
\begin{figure}
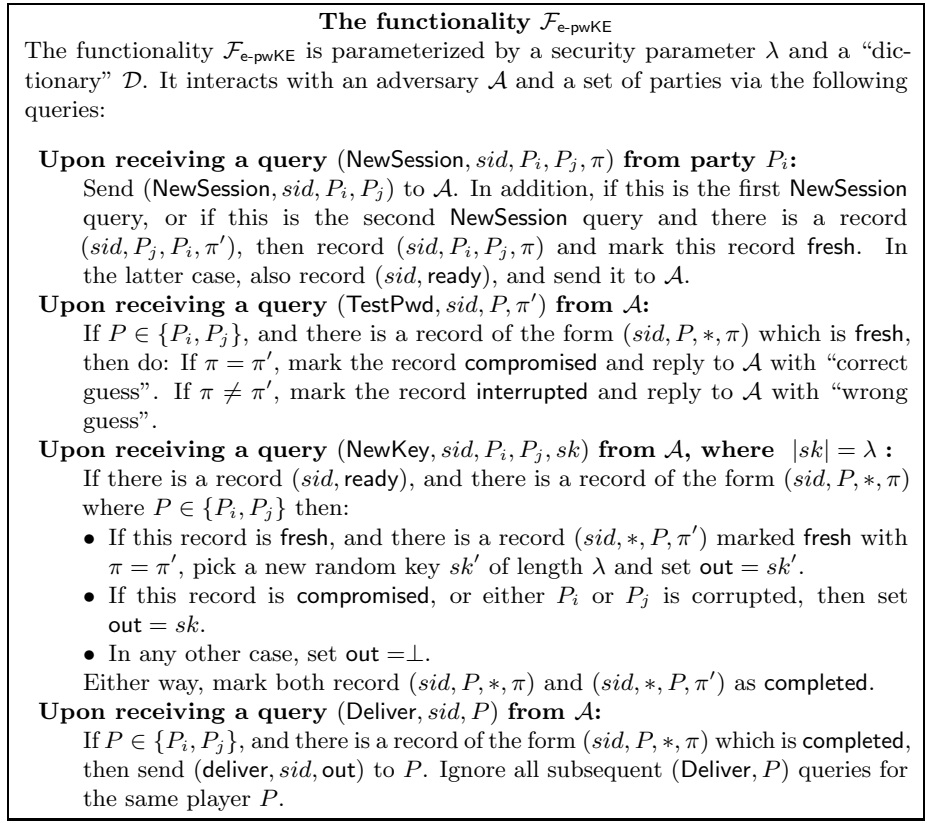

\begin{tcolorbox}[colback=white,size=small,sharp corners,colframe=black]
{\centering
	\textbf{The functionality} $\ideal{e-pwKE}$ \\
}

The functionality $\ideal{e-pwKE}$ is parameterized by a security parameter
$\lambda$ and a ``dictionary'' $\mathcal{D}$.
It interacts with an adversary $\adv{A}$ and a set of parties via the
following queries:
\begin{itemize}[leftmargin=*,label={}]
\item \textbf{Upon receiving a query} $(\textsf{NewSession}, sid, P_i, P_j,
\pi)$ \textbf{from party $P_i$:}
\begin{itemize}[label={}]
\item Send $\left( \textsf{NewSession}, sid, P_i, P_j \right) $ to $\adv{A}$. In addition, if this is the first
$\textsf{NewSession}$ query, or if this is the second $\textsf{NewSession}$
query and there is a record $\left(sid, P_j, P_i, \pi' \right) $, then record
$\left(sid, P_i, P_j, \pi \right) $ and mark this record $\textsf{fresh}$.
In the latter case, also record $(sid, \textsf{ready})$, and send it to
$\adv{A}$.
\end{itemize}
\item \textbf{Upon receiving a query} $(\textsf{TestPwd}, sid, P, \pi')$
\textbf{from $\adv{A}$:}
\begin{itemize}[label={}]
\item If $P \in \{P_i, P_j\}$, and there is a record of the form
$(sid, P, \ast, \pi)$ which is \textsf{fresh},
then do: If $\pi = \pi'$, mark the record
$\textsf{compromised}$ and reply to $\adv{A}$ with ``correct guess''. If $\pi
\neq \pi'$, mark the record $\textsf{interrupted}$ and reply to $\adv{A}$
with ``wrong guess''.
\end{itemize}
\item \textbf{Upon receiving a query} $(\textsf{NewKey}, sid, P_i, P_j, sk)$ \textbf{from} $\adv{A}$\textbf{, where } $\size{sk} = \lambda$ \textbf{:}
\begin{itemize}[label={}]
\item If there is a record $\left(sid, \textsf{ready} \right)$, and there is a record of the form $(sid, P, \ast, \pi)$ where
$P \in \{P_i, P_j\}$ then:
\begin{itemize}[leftmargin=*,label={$\bullet$}]
\item If this record is $\textsf{fresh}$, and there is a record $(sid, \ast, P,
\pi')$ marked $\textsf{fresh}$ with $\pi = \pi'$, pick a new random key $sk'$
of length $\lambda$ and set $\textsf{out} = sk'$.
\item If this record is $\textsf{compromised}$, or either $P_i$ or $P_j$ is
corrupted, then set $\textsf{out} = sk$.
\item In any other case, set $\textsf{out} = \perp$.
\end{itemize}
\item Either way, mark both record $\left(sid, P, \ast, \pi \right) $ and
$(sid, \ast, P, \pi')$ as $\textsf{completed}$.
\end{itemize}
\item \textbf{Upon receiving a query} $(\textsf{Deliver}, sid, P)$
\textbf{from $\adv{A}$:}
\begin{itemize}[label={}]
\item If $P \in \{P_i, P_j\}$, and there is a record of the form
$(sid, P, \ast, \pi)$ which is \textsf{completed},
then send $(\textsf{deliver}, sid, \textsf{out})$
to $P$. Ignore all subsequent $(\textsf{Deliver}, P)$ queries for the same
player $P$.
\end{itemize}

\end{itemize}
\end{tcolorbox}
\caption{The password-based key-exchange functionality $\ideal{e-pwKE}$ with
explicit mutual authentication.}
\label{fig:f_epwKE}
\end{figure}
}

\subsubsection{$\ideal{e-pwKE}$ implies $\ideal{EQ}$.}
\label{subsec:pid-equality}

We define an equality-testing functionality $\ideal{EQ}$ (\Cref{fig:f_eq})
that, roughly speaking, takes inputs from two parties and does the following:

\begin{itemize}
\item if the inputs are equal, outputs the value $1$ to both parties;
moreover, if either party is corrupted, the adversary is allowed to set the
output.
\item if the inputs are unequal, send both parties the special symbol $\perp$.
\end{itemize}

More precisely, $\ideal{EQ}$ captures a protocol between two parties $P_i, P_j$
started by having the two parties sending messages to the
functionality with their secret strings $\pi_i, \pi_j$.
If the inputs match, the functionality assigns the output to be~$1$,
otherwise
it sets the output to be $\perp$.
Finally, the adversary $\adv{A}$ instructs the functionality when to send the
output to both parties.
Thus, this definition corresponds to achieving \emph{explicit mutual
authentication}.
We also allow the adversary three special powers.
First, we allow him to set the output if one of the parties is corrupted and
both the parties have the same input.
Furthermore, he controls the delivery of messages to the parties.
This is an ability that he inevitably has in the real world.
Finally, as in the case of PAKE, the low entropy of the messages in the
dictionary $\mathcal{D}$ makes online dictionary attacks unavoidable, which is
captured by the \textsf{Test} query given to the adversary.

\begin{figure}[!h]
\begin{tcolorbox}[colback=white,size=small,sharp corners,colframe=black]
{\centering
	\textbf{The functionality} $\ideal{EQ}$ \\
}

The functionality $\ideal{EQ}$ is parameterized by a security parameter
$\lambda$ and a ``dictionary'' $\mathcal{D}$.
It interacts with two parties $P_i, P_j$, and an adversary $\adv{A}$ via the
following queries:
\begin{itemize}[leftmargin=*,label={}]
\item \textbf{Upon receiving a query} $(\textsf{NewSession}, sid, P_i, P_j,
\pi)$ \textbf{from party $P_i$:}
\begin{itemize}[label={}]
\item Send $\left( \textsf{NewSession}, sid, P_i, P_j \right) $ to $\adv{A}$.
In addition, do the following:
\begin{itemize}[leftmargin=*,label={$\bullet$}]
\item If this is the first $\textsf{NewSession}$ query, then record
$\left(sid, P_i, P_j, \pi \right) $ and mark this record $\textsf{fresh}$.
\item If this is the second $\textsf{NewSession}$
query and there is a record $\left(sid, P_j, P_i, \pi' \right) $  which is
$\textsf{fresh}$, then do: if $\pi = \pi'$, then set $\textsf{out} = 1$, otherwise, set $\textsf{out} = \perp$. Mark both records $\textsf{completed}$.
\end{itemize}
\end{itemize}

\item \textbf{Upon receiving a query} $(\textsf{Test}, sid, P, \pi'), P \in \{P_i, P_j\}$ \textbf{from} $\adv{A}$ \textbf{:}
\begin{itemize}[label={}]
\item If there is a record of the form
$(sid, P, \ast, \pi)$ which is \textsf{fresh},
then do: If $\pi = \pi'$, mark the record
$\textsf{compromised}$ and reply to $\adv{A}$ with ``correct guess''. If $\pi
\neq \pi'$, mark the record $\textsf{interrupted}$ and reply to $\adv{A}$
with ``wrong guess''.
\end{itemize}
\item \textbf{Upon receiving a query} $(\textsf{Output}, sid, \gamma), \gamma \in \{1, \perp\}$ \textbf{from} $\adv{A}$ \textbf{:}
\begin{itemize}[label={}]
\item If there is a record of the form $(sid, \ast, \ast, \pi)$ which is
\textsf{compromised}, or one of the parties is corrupted, then set
$\textsf{out} = \gamma$. If this record is $\textsf{interrupted}$,
then set $\textsf{out} = \perp$. Otherwise, do nothing.
\end{itemize}
\item \textbf{Upon receiving a query} $(\textsf{Deliver}, sid, P),
P \in \{P_i, P_j\}$ \textbf{from} $\adv{A}$ \textbf{:}
\begin{itemize}[label={}]
\item If there is a record of the form $(sid, P, \ast, \pi)$
which is \textsf{completed}, send $(\textsf{deliver}, sid,
\textsf{out})$ to the player $P$.
Ignore all subsequent $(\textsf{Deliver}, P)$ queries for the same player $P$.
\end{itemize}
\end{itemize}
\end{tcolorbox}
\caption{The equality-testing functionality $\ideal{EQ}$.}
\label{fig:f_eq}
\end{figure}

The following lemma shows that the $\ideal{e-pwKE}$ functionality
already implements the $\ideal{EQ}$. Though this seems to be a folklore, we
also give a proof of this lemma for completeness.

\begin{lemma}\label{lemma:PAKE-EQUALITY}
There is a protocol that perfectly implements the $\ideal{EQ}$ functionality in
the $\ideal{e-pwKE}$ hybrid model, tolerating adaptive corruptions and without
assuming authenticated channels.
\end{lemma}

\begin{proof}
The protocol that implements $\ideal{EQ}$ simply forwards the
parties' messages to the $\ideal{e-pwKE}$ functionality.
In particular, on input $(sid, \pi_i)$ from the environment, the party
$P_i$ sends a message $(\textsf{NewSession}, sid, P_i, P_j, \pi)$ to
$\ideal{e-pwKE}$.
When $P_i$ receives a message $(\textsf{deliver}, sid, \textsf{out})$ back from
$\ideal{e-pwKE}$, if $\textsf{out} \neq \perp$, $P_i$ outputs $1$, otherwise,
it outputs $\perp$ and terminates. Similarly, $P_j$ does the same.

We simply show how to simulate the adversary $\adv{A}'s$ messages.
\begin{itemize}[align=left]
\item[\textbf{Simulating a} $(\textsf{Test}, sid, P, \pi)$ \textbf{query from
$\adv{A}$:}] If $\adv{A}$ already sent a $(\textsf{Deliver}, P)$ query
before, ignore this query.
Otherwise, send a query $(\textsf{TestPwd}, sid, P, \pi)$
to $\ideal{e-pwKE}$, and record the response from $\ideal{e-pwKE}$ (either
``correct guess'' or ``wrong guess'').
\item[\textbf{Simulating a} $(\textsf{Output}, sid, \gamma)$ \textbf{query
from $\adv{A}$:}] If $\adv{A}$ already sent a $(\textsf{Deliver}, P)$ query
before, ignore this query.
Otherwise, send a query $(\textsf{NewKey}, sid, P_i, P_j, \gamma)$
to $\ideal{e-pwKE}$.
\item[\textbf{Simulating a} $(\textsf{Deliver}, sid, P)$ \textbf{query
from $\adv{A}$:}] If $\adv{A}$ already sent a $(\textsf{Deliver}, P)$ query
before, ignore this query. Otherwise, send a query $(\textsf{Deliver}, sid, P)$
to $\ideal{e-pwKE}$.
\end{itemize}

It is easy to see that the simulation is perfect, and the view of the
environment is identical in the real execution of $\adv{A}$
in the protocol (in the $\ideal{e-pwKE}$-hybrid model)
and the simulated ideal-model execution with $\ideal{EQ}$.
\end{proof}

\subsection{Proof of
\Cref{theorem:no_sim_pake}}\label{subsec:impossibility-proofs:SB}



To prove Theorem~\ref{theorem:no_sim_pake}, we employ a general result which
proves that for the class of deterministic, two-sided functionalities
including the equality-testing function, the
security for one party implies complete insecurity for the other in the
simulation-based model.

\begin{lemma}[{\cite[Theorem 2]{buhrman2012complete}}]
\label{lemma:BCS12}
If a protocol $\pi$ for the evaluation of a deterministic two-sided function
$F$ is $\varepsilon$-correct
and $\varepsilon$-secure against Bob, then there is a cheating strategy for
Alice (where she uses input $u_0$ and Bob has input $v$) which gives her
$\tilde{v}$ distributed according to some distribution $Q(\tilde{v}|u_0, v)$
such that for all $u$: $\pr{\tilde{v} \gets Q: F(u, v) = F(u, \tilde{v})}
\geq 1 - 28\varepsilon$.
\end{lemma}


\begin{proof}[Proof (Theorem \ref{theorem:no_sim_pake})]
First we note that the reduction from $\ideal{EQ}$ to $\ideal{e-pwKE}$ in
\Cref{lemma:PAKE-EQUALITY}
holds unconditionally in the UC
model, which
implies perfect security in the simulation-based model.
We then prove by contradiction, if there is a
statistically secure PAKE protocol in the plain model,
then by \Cref{lemma:PAKE-EQUALITY}, that protocol is also a statistically
secure protocol for $\ideal{EQ}$ in the plain model,
which violates \Cref{lemma:BCS12}.
\end{proof}



\subsection{Proof of
\Cref{theorem:no_uc_pake}}\label{subsec:impossibility-proofs:UC}


First note that according to the following lemma, the impossibility of
everlasting quantum-UC security implies the impossibility
of statistical quantum-UC security.

\begin{lemma}[{\cite[Lemma 1]{C:Unruh13}}]
\label{lemma:ever_stat}
Let $\pi$ and $\rho$ be protocols. If $\pi$ statistically
quantum-UC-emulates~$\rho$, then $\pi$ everlastingly
quantum-UC-emulates~$\rho$.
\end{lemma}

In the following, we thus focus on the proof for the everlasting security.

Assuming some trusted setup, the following lemma states the impossibility of
everlastingly realizing $\ideal{EQ}$ using only quantum-passively-realizable
functionalities, including $\ideal{CRS}$ (described in
\Cref{fig:ideal_crs}).

\begin{lemma}
\label{lemma:no_equality}
There is no statistically or everlastingly quantum-UC
secure protocol that realizes $\ideal{EQ}$ which only uses
quantum-passively-realizable functionalities as trusted setup assumptions.
\end{lemma}

Before proving \Cref{lemma:no_equality}, we recall the
impossibility of achieving everlastingly quantum-UC-secure oblivious transfer.
\begin{lemma}[{\cite[Theorem 5]{C:Unruh13}}]
There is no statistically or everlastingly quantum-UC
secure OT protocol which only uses quantum-passively-realizable
functionalities as trusted setup assumptions.
\label{theorem:no_ot}
\end{lemma}


We use the notion of reductions between MPC functionalities, that allows us to
form ``classes'' of functionalities with similar cryptographic complexity:
Following~\cite{C:MajPraRos10}, a functionality is said \emph{trivial} or
\emph{feasible} if it can be realized in the UC framework in the plain model
(with no setup assumptions), and it is said \emph{complete} if it is
sufficient for computing arbitrary other functions, under appropriate
complexity assumptions, when used as trusted setups. We recall the following
results that are proven in \cite{EC:Unruh10,TCC:FKSZZ13}.

%

\begin{lemma}[{\cite[Theorem 15]{EC:Unruh10} and
\cite[Theorem 2]{TCC:FKSZZ13}}]
\label{lemma:quantum_lifting}
The following statements hold:
\begin{enumerate}
\item If a protocol $\pi$ statistically UC realizes a functionality
$\ideal{}$, then $\pi$ statistically quantum-UC realizes the functionality
$\ideal{}$ (Quantum lifting theorem).
\item Feasibility in the quantum world is \emph{equivalent} to classical
feasibility, in both the computational and statistical
setting.
\end{enumerate}
\end{lemma}

To show a reduction from $\ideal{EQ}$ to $\ideal{OT}$, we employ the
following intermediate results.

\begin{definition}[OT-cores] Let $F$ be a deterministic two-party function,
$\Gamma_A$, $\Gamma_B$ be the input alphabet of two parties, $\Omega_A$,
$\Omega_B$ be the output distribution of two parties, and $f_A, f_B$ is the
output values of the two parties. A quadruple $(x, x^\prime, y, y^\prime) \in
\Gamma_A^2 \times \Gamma_B^2$ is an \emph{OT-core} of $F$, if the following
three conditions are met:
\begin{enumerate}
\item We have that $f_A(x, y) = f_A(x, y^\prime)$.
\item We have that $f_B(x, y) = f_B(x^\prime, y)$.
\item We have that $f_A(x^\prime, y) \neq f_A(x^\prime, y^\prime)$ or $f_B(x,
y^\prime) \neq f_B(x^\prime, y^\prime)$ (or both).
\end{enumerate}
\end{definition}
In \cite{TCC:KraMul11} the so-called \emph{Classification theorem} was proven,
which shows a necessary and sufficient condition to have a
reduction protocol from an ideal functionality $\ideal{}$ to $\ideal{OT}$.
\begin{theorem}[{The Classification Theorem~\cite{TCC:KraMul11}}]
\label{theorem:classification}
There exists an OT protocol that is statistically secure against passive
adversaries in the $\ideal{}$-hybrid model, for some $\ideal{}$, if and only
if $\ideal{}$ has an OT-core.
\end{theorem}

\begin{proof}[Proof of \Cref{lemma:no_equality}]




We first show that the equality-testing function
$\mathcal{F}_{\textsf{EQ}}$ admits an OT-core.
Consider $\mathcal{F}_{\textsf{EQ}} \coloneqq (\Gamma_A,
\Gamma_B, \Omega_A,$ $\Omega_B, f_A, f_B)$, without loss of generality, assume
$\Gamma_A = \Gamma_B = \Gamma$. Let $c \in \Gamma$ be a random value drawn
from the input distribution, then a quadruple $(c, c+1, c-1, c+1)$ is an
OT-core of $\mathcal{F}_{\textsf{EQ}}$ because:
\begin{align*}
f_A(c, c-1) &= f_A(c, c+1) = 0 \\
f_B(c, c-1) &= f_B(c+1, c-1) = 0 \\
0 = f_A(c+1, c-1) &\neq f_A(c+1, c+1) = 1
\end{align*}

Then the classification theorem (\Cref{theorem:classification}) tells us that
there exists an OT protocol that is statistically secure against passive
adversaries in the $\ideal{EQ}$-hybrid model. Using the lifting theorem
(\Cref{lemma:quantum_lifting}), that protocol is also statistically secure
against quantum-passive adversaries in the $\ideal{EQ}$-hybrid model.

We now prove the lemma by contradiction.
Assume that there exists an everlasting quantum-UC-secure protocol
$\pi$ realizing $\ideal{EQ}$ which only uses quantum-passively-realizable
functionalities.
Let $\rho$ be the protocol resulting from $\pi$ by replacing
invocations of $\ideal{EQ}$ by invocations of the subprotocol $\pi$.
Then $\rho$ is a everlasting quantum-UC-secure protocol realizing $\ideal{OT}$
which only uses quantum-passively-realizable functionalities against
quantum-passive adversaries.
This contradicts \Cref{theorem:no_ot}.

Because of \Cref{lemma:ever_stat}, the impossibility of
statistical security follows immediately from the impossibility of everlasting
security.
\end{proof}

\begin{proof}[Proof of \Cref{theorem:no_uc_pake}]
Similarly to \Cref{theorem:no_sim_pake}, the result follows from
\Cref{lemma:PAKE-EQUALITY} and \Cref{lemma:no_equality}.
\end{proof}



\section{Proof of Theorem \ref{theorem:compiler}}
\label{app:proof_compiler}
\begin{proof}
We consider a protocol $\Phi = \mathcal{C}^{\varepsilon_{sig}}(\Pi)$ which is
the protocol resulting by applying the split authentication transformation to
a two-party quantum protocol $\Pi$ between two honest parties $P_1, P_2$ and
an adversary $\adv{A}$. The simulator of protocol $\Pi$ is denoted by
$\mathcal{S}_\Pi$.

\paragraph{Simulating when neither of the parties is corrupted: } We show that
the security of $\Phi$ in this case reduces to either the security of $\Pi$
when one of the parties is corrupted or the security of $\Pi$ when neither of
the parties is corrupted in the presence of a quantum passive unbounded
adversary. Formally, we say that $\textsf{vk}_j$ is $P_j$'s \emph{authentic
key} if it is the key generated by $P_j$ in the internal simulation by
$\adv{S}$. Simulator $\adv{S}$ internally invokes a copy of two uncorrupted
parties and runs an interaction between $\adv{A}$ and these simulated copies
as follows:
\begin{enumerate}
\item Whenever $\adv{A}$ delivers a message $\textsf{vk}$ to an uncorrupted
party $P_i$, $\adv{S}$ simulates the actions of $P_i$ in the link
initialization phase.
\item Whenever an internally simulated uncorrupted party $P_i$ completes the
link initialization phase with output $\textsf{sid}$, simulator $\adv{S}$
determines the set $H_i$ to be the set of the uncorrupted party $P_j$ such
that the authentic verification key sent by $P_j$ is included in
$\textsf{sid}$. (Recall that because $\adv{S}$ internally runs all uncorrupted
parties, it can determine whether or not the key $\textsf{vk}_j$ that is
chosen for $P_j$ is included in $\textsf{sid}$). $\adv{S}$ then checks for the
previously computed set $H_j$, it holds that either:
\begin{itemize}
\item $H_i = H_j \neq \varnothing$ and $\textsf{sid}_{H_i} =
\textsf{sid}_{H_j}$, or
\item $H_i = H_j = \varnothing$ and $\textsf{sid}_{H_i} \neq
\textsf{sid}_{H_j}$.
\end{itemize}
If this holds, then $\adv{S}$ runs the simulator $\mathcal{S}_\Pi$. Otherwise,
$\adv{S}$ halts and outputs $\textsf{fail}$.
\item $\adv{S}$ outputs whatever $\adv{A}$ outputs.
\end{enumerate}

It is easy to verify that as long as $\adv{S}$ does not output
$\textsf{fail}$, the security of $\Phi$ reduces to the security of $\Pi$: the
first condition corresponds to the security of $\Pi$ with quantum passive
adversaries, and the second condition corresponds to the security of $\Pi$ in
the case one of the parties is corrupted. It therefore suffices to show that
$\adv{S}$ outputs a $\textsf{fail}$ with at most negligible probability
$\varepsilon_{sig}$.

There are three events that could cause a $\textsf{fail}$ message:
\begin{enumerate}
\item \emph{$H_i = H_j$, and yet $\textsf{sid}_{H_i} \neq \textsf{sid}_{H_j}$:}
By the behavior of $\adv{S}$, we have that if $H_i$ and $H_j$ are defined, and
$H_i = H_j \neq \varnothing$, then $P_i$ received $P_j$'s authentic key and
vice versa. On the other hand, $P_j$ only concludes this phase with output if
$\textsf{sid}_i = \textsf{sid}_j$ and if the verification of $\textsf{sid}_i$
with the verification key of $P_i$ passes. If the event we are considering
here occurred with a non-negligible probability, then $\textsf{sid}_i \neq
\textsf{sid}_j$, and so in the internal simulation by $\adv{S}$, we have that
$P_i$ has never signed on the \uccmd{sid} which $P_j$ received in the name of
$P_i$. Thus, $\adv{A}$ must have forged a signature, and it can be used to
break the signature scheme.
\item \emph{$H_i \neq \varnothing$, and $H_j = \varnothing$:} Let $H_i =
\{P_j\}$ with $P_j$ be an uncorrupted party. Then, using the same arguments as
above, except with negligible probability, $P_i$ must have the same
$\textsf{sid}$ as $P_j$. By the construction of $\adv{S}$, it therefore holds
that $H_i = H_j$.
\item \emph{$H_i = H_j = \varnothing$, and yet $\textsf{sid}_i =
\textsf{sid}_j$:} By the construction of $\adv{S}$, if $\textsf{sid}_i =
\textsf{sid}_j$ then $H_i = H_j$. This event therefore never occurs.
\end{enumerate}

We conclude that $\Phi$ is everlastingly
$\varepsilon+\varepsilon_{sig}$-secure according to Definition
\ref{def:2pc_everlasting} in this case.

\paragraph{Simulating when one of the parties is corrupted: } Since the
protocol is completely symmetric between the two parties, the simulation for a
corrupted party is identical to that for the other corrupted party.

Simulator $\adv{S}$ internally invokes a copy of the uncorrupted party and
runs an interaction between $\adv{A}$ and the simulated copy as follows:

\begin{enumerate}
\item In the \emph{Link Initialization} phase, $\adv{S}$ behaves honestly and
aborts if $\adv{A}$ aborts.
\item After the \emph{Link Initialization} phase, $\adv{S}$ runs the simulator
$\mathcal{S}_\Pi$.
\item $\adv{S}$ outputs whatever $\adv{A}$ outputs.
\end{enumerate}

It is straightforward to verify that a real execution of protocol $\Phi$ is
identical to its ideal-model execution, and that the security of $\Phi$
reduces to the security of $\Pi$. This is due to the fact that $\adv{S}$ just
mimics the actions of the uncorrupted party and the local outputs of the
uncorrupted party in the internal simulation correspond exactly to the outputs
of the actual uncorrupted party in the ideal model. Thus $\Phi$ is
everlastingly $\varepsilon$-secure in this case.
\end{proof}

\section{Proof of Theorem \ref{theorem:main_theorem_1}}
\label{app:proof_main_theorem_1}

\subsection{Technical Tools}
Before proceeding through the actual proof, we recall
some technical tools.

\subsubsection{Conditional Independence.} We need to express that a
random variable $X$ is independent of a quantum state $E$ \emph{when given a
random variable} $Y$. Independence means that when given $Y$, the state $E$
gives no additional information on $X$. Another way to understand this is
that $E$ can be obtained from $X$ and $Y$ by solely processing $Y$. Formally,
adopting the notion introduced in \cite{C:DFSS07b}, this is expressed by
requiring that $\rho_{XYE}$ equals $\markovstate{X}{Y}{E}$, where the latter
is defined as
\[\markovstate{X}{Y}{E} \coloneqq \sum_{x, y}P_{XY}\left( xy
\right)\ketbra{x}{x}\otimes \ketbra{y}{y} \otimes \rho_E^{y}. \] In other
words, $\rho_{XYE} = \markovstate{X}{Y}{E}$ precisely if $\rho_E^{xy} =
\rho_E^y$ for all $x$ and $y$. To further illustrate its meaning, notice that
if the $Y$-register is measured and value $y$ is obtained, then the state
$\markovstate{X}{Y}{E}$ collapses to $\left( \sum\nolimits_{x}P_{X|Y}\left(
x|y \right) \ketbra{x}{x}\right) \otimes \rho_E^y$, so that indeed no further
information on $x$ can be obtained from the $E$-register. This notation
naturally extends to $\markovstate{X}{Y}{E|\mathcal{E}} \coloneqq
\sum\nolimits_{x, y}P_{XY|\mathcal{E}}\left( xy \right)\ketbra{x}{x}\otimes
\ketbra{y}{y} \otimes \rho_{E|\mathcal{E}}^{y}.$

\subsubsection{Technical Lemmas.}
The following chain rule shows that the conditional min-entropy
$\sminentro{A|B}_{\rho}$ can decrease by at most $\log{\size{Z}}$ when
conditioning on an additional classical system Z.
\begin{lemma}[{\cite[Lemma 11]{winkler2011impossibility}}]
\label{lemma:chain_rule}
Let $\varepsilon \geq 0$, and let $\rho_{ABZ}$ be a tripartite state that is
classical on $Z$ with respect to some orthonormal basis $\left\{ \ket{z}
\right\}_{z\in \mathcal{Z}}$. Then
\[\sminentro{A|BZ}_{\rho} \geq \sminentro{A|B}_{\rho} - \log{\size{Z}}.\]
\end{lemma}

We consider a tri-partite quantum state
$\rho_{ABC}$, and two generalized measurements acting on
$A$: $\mathbb{X}$ with elements $\left\{ M_A^x \right\}$ and $\mathbb{Z}$ with
elements $\left\{ N_A^z \right\}$. The joint state of the classical outcome $X$
when measuring $A$ with respect to $\mathbb{X}$ and the system B is given as a
bipartite cq-state $\rho_{XB} \coloneqq \sum_{x} \ketbra{x}{x} \otimes
\tau_{B}^x, \enskip
\textnormal{where } \tau_B^x =  \textnormal{tr}_{AC}\left\{
{M_A^x} ^{\dagger} M_A^x \rho_{ABC}
\right\}$. Similarly, we define $\rho_{ZC}$, where the measurement
$\mathbb{Z}$ instead of $\mathbb{X}$ is applied to $A$ and where we keep
system $C$ instead of $B$. Assume $\textnormal{dim}\left( \hilbert{A} \right)
= n$, we define
\[c_i \coloneqq \underset{x, z}{\textnormal{max}} \left\| M^{+,x}_{A_i}
\left( M^{\times,z}_{A_i} \right)^{\dagger} \right\|_\infty^2,
\textnormal{and}\; \bar{c} \coloneqq \underset{i \in
\intval{n}}{\textnormal{max}}\left(\prod c_i \right)^{\frac{1}{n}}. \]
We state the uncertainty relation in a form of smooth min- and max-entropy,
applying to the setup with the computational and Hadamard basis.
\begin{theorem}[{\cite[Theorem 1]{tomamichel2011uncertainty}}]
\label{theorem:uncertainty_tradeoff}
Let $\rho_{ABC} \in \mathcal{P}\left( \hilbert{A} \otimes \hilbert{B} \otimes
\hilbert{C} \right) $, let $\varepsilon > 0$ and let $\mathbb{X}$ and
$\mathbb{Z}$ be two generalized measurements on $A$. Then,
\[\sminentro{X|B}_{\rho} + \smaxentro{Z|C}_{\rho} \geq
\log{\frac{1}{\bar{c}}}n,\] where $\bar{c} \in \left( 0, 1 \right)$.
\end{theorem}

We now complete the proof of Theorem~\ref{theorem:main_theorem_1} in the
following sections. In our proof, we use upper case letters for the random
variables in the proofs that describe the respective values in the protocol.
In particular, we write $W, X_W, \hat{X}_W$ for the random variables taking
values $pw, x|_{I_w}, \hat{x}|_{I_w}$, respectively.

In our proof, we assume that the (quantum) system containing all the
information a potential adversary might have gained during the protocol
execution can be decomposed into a classical part $Z$ and a purely quantum
part $E$. Because the commitment scheme is perfectly hiding, it essentially
leaks no information, thus we omit the transcript of the commitments in the
description of $Z$ (in other words, $Z$ implicitly includes the transcript of
the commitments). We write $Z = (Z', S, J)$ and understand that $Z'$ denote
the classical system of the adversary without the random variables $S$ and
$J$.

\subsection{Simulating the case when neither of the parties is Corrupted}
\label{proof:no_corrupted}
In order to show that the protocol is secure, it suffices to show that
\[\td{\rho_{K_{P_1}K_{P_2}WZE}}{\markovstate{W}{Z}{E}} = \frac{1}{2}\left\|
\rho_{K_{P_1}K_{P_2}WZE} - \rho_{K_{P_1}K_{P_2}} \otimes
\markovstate{W}{Z}{E}\right\|_1 \] is negligible, where
$\rho_{K_{P_1}K_{P_2}WZE}$ is the common output state of the protocol and
$\rho_{K_{P_1}K_{P_2}}$ is defined as a perfect key as follows:
\[\rho_{K_{P_1}K_{P_2}} \coloneqq \frac{1}{2^\lambda}\sum_{\textsf{sk} \in
\bin^\lambda}\ketbra{\textsf{sk}}{\textsf{sk}}\otimes
\ketbra{\textsf{sk}}{\textsf{sk}}.\]

The proof will be completed using the following claims, which are proven
below:
\begin{enumerate}
\item \textit{Claim \ref{claim:proof_1_correctness}.} The correctness of the
protocol:
\[\pr{K_{P_1} \neq K_{P_2}|K_{P_1} \neq \perp, K_{P_2} \neq \perp} \leq
\varepsilon_{cor}.\]
\item \textit{Claim \ref{claim:proof_1_secrecy}.} The secrecy of the session
key:
\[\dis{\rho_{K_{P_1}K_{P_2}WZE}} \leq \varepsilon_{sec}.\]
\item \textit{Claim \ref{claim:proof_1_independence}.} $W$ is independent of
the adversary's quantum system:
\[\markovstate{W}{Z}{E} = \rho_{WZE}.\]
\end{enumerate}

We now complete the proof as follows. Similarly to QKD's proof, we will employ
the following lemma which allows us to split the norm into two terms
corresponding to correctness and secrecy.

\begin{lemma}[{\cite[Lemma 1]{Tomamichel2017largelyself}}] Let
$\varepsilon_{cor},\varepsilon_{sec} \in [0, 1)$ be two constants. If, for
every common input state $\rho_{ABE} \in \mathcal{P}\left( \hilbert{A} \otimes
\hilbert{B} \otimes \hilbert{E} \right) $ and $\rho_{K_{P_1}K_{P_2}WZE} = \Pi
\left( \rho_{ABE} \right) $, we have
\[\pr{K_{P_1} \neq K_{P_2}|K_{P_1} \neq \perp, K_{P_2} \neq \perp} \leq
\varepsilon_{cor}\] and
\[\dis{\rho_{K_{P_1}WZE}} \leq \varepsilon_{sec}.\] Then,
$\dis{\rho_{K_{P_1}K_{P_2}WZE}} \leq \varepsilon_{cor} + \varepsilon_{sec}$.
\end{lemma}

The following sequence of hybrids establishes what we want, where the last
inequality follows from the above lemma.
\begin{align*}
\td{\rho_{K_{P_1}K_{P_2}WZE}}{\markovstate{W}{Z}{E}} =&\enskip \frac{1}{2}\left\| \rho_{K_{P_1}K_{P_2}WZE} - \rho_{K_{P_1}K_{P_2}} \otimes \markovstate{W}{Z}{E}\right\|_1 \\
= &\enskip \frac{1}{2}\left\| \rho_{K_{P_1}K_{P_2}WZE} - \rho_{K_{P_1}K_{P_2}} \otimes \rho_{WZE}\right\|_1 \\
= &\enskip \dis{\rho_{K_{P_1}K_{P_2}WZE}} \\
\leq &\enskip \varepsilon_{cor} + \varepsilon_{sec}.
\end{align*}

It remains to prove the three claims, which essentially give bounds on the
security parameters in terms of the protocol parameters, made in the proof
above.

The first claim establishes correctness of the protocol. Correctness of the
protocol is ensured in the error correction step using private error
correction and consequently correctness can be bounded in terms of the
probability of failure decoding of the small-bias family of codes.
\begin{claim}
\label{claim:proof_1_correctness}
Let $\varepsilon_{cor}$ be the probability of failure decoding of binary
linear code $C$. For every common input state $\rho_{AB} \in
\mathcal{P}\left(\hilbert{A} \otimes \hilbert{B} \right) $ and
$\rho_{K_{P_1}K_{P_2}WZ} = \Pi \left( \rho_{AB} \right) $ we have
\[\pr{K_{P_1} \neq K_{P_2}|K_{P_1} \neq \perp, K_{P_2} \neq \perp} \leq
\varepsilon_{cor}.\]
\end{claim}
\begin{proof}
We consider the following chain of inequalities:
\begin{align*}
\pr{K_{P_1} \neq K_{P_2}|K_{P_1} \neq \perp, K_{P_2} \neq \perp} & = \pr{F(X_W) \neq F(\tilde{X}_W) \wedge \fname{decode}\left( S \right) = \tilde{X}_W } \\
& \leq \pr{X_W \neq \tilde{X}_W \wedge \fname{decode}\left( S \right) = \tilde{X}_W } \\
& = \pr{X_W \neq \tilde{X}_W} \cdot \pr{\fname{decode}\left( S \right) = \tilde{X}_W|\tilde{X}_W \neq X_W} \\
& \leq \pr{\fname{decode}\left( S \right) = \tilde{X}_W| \tilde{X}_W \neq X_W} \leq \varepsilon_{cor}.
\end{align*}

Note that $S = \fname{synd}\left( X_W \right) $. The first inequality is a
consequence of the fact $X_W = \tilde{X}_W$ implies $F(X_W) = F(\tilde{X}_W)$.
The second inequality follows since $\pr{X_W \neq \tilde{X}_W} \leq 1$ and the
last one by definition of the error-correcting code.
\end{proof}

The second claim asserts secrecy.
\begin{claim}
\label{claim:proof_1_secrecy}
For every common input state $\rho_{ABE} \in \mathcal{P}\left( \hilbert{A}
\otimes \hilbert{B} \otimes \hilbert{E} \right) $ and $\rho_{K_{P_1}K_{P_2}WZE}
= \Pi \left( \rho_{ABE} \right) $ we have
\begin{align*}
\dis{\rho_{K_{P_1}WZE}} \leq \varepsilon_{ec} + \varepsilon_{pa} =
2^{-\frac{1}{2}\left(g(\varepsilon) + \frac{\beta n}{2}\right)} + \left(2
\varepsilon + 2^{-\frac{1}{2}\left( g(\varepsilon) - \lambda \right) } \right),
\end{align*}
for some $\varepsilon > 0$ and $\bar{c} \in \left( 0, 1 \right) $, where
$g(\varepsilon)$ is given as $g(\varepsilon) =
\left(\log{\frac{1}{\bar{c}}} - h \left( \tau + \varepsilon
\right) - \frac{1}{2}\right)n$.
\end{claim}

Instead of following the proof of statistical bounds on the min-entropy and
max-entropy of QKD's proof, we leverage the proof technique from
\cite{C:DFLSS09}.

In the following, let $T = T_1 \cup T_2$ and $T' = \left\{ i \in T \mid
\theta_i = \hat{\theta}_i \right\}$, is a random subset of arbitrary size of
$T$. Let the random variable $Test$ describe the choice of $test = \left( T,
T' \right) $ as specified above, and consider the state
\[\rho_{TestAE} = \rho_{Test} \otimes \ketbra{\varphi_{AE}}{\varphi_{AE}} =
\sum_{test}P_{Test} \left( test \right)\ketbra{test}{test} \otimes
\ketbra{\varphi_{AE}}{\varphi_{AE}} \] consisting of the classical $Test$ and
the quantum state $\varphi_{AE}$.

\begin{lemma}[{\cite[Corollary 4.4]{C:DFLSS09}}]
\label{lemma:proof_1_bounding_entropy}
For any $\varepsilon > 0$, $\hat{x} \in \bin^m$, and $\hat{\theta} \in
\bases^m$, define
\[\tilde{\rho}_{TestAE} = \sum_{test}P_{Test} \left( test
\right)\ketbra{test}{test} \otimes
\ketbra{\tilde{\varphi}_{AE}^{test}}{\tilde{\varphi}_{AE}^{test}}, \] where
for any $test = \left( T, T' \right) $:
\[\ket{\tilde{\varphi}_{AE}^{test}} = \sum_{x \in B_{test}}
\alpha_x^{test}\ket{x}_{\hat{\theta}}\ket{\psi}_E^x\] for $B_{test} = \left\{
x \in \bin^m \mid r_H \left( x|_{\bar{T}}, \hat{x}|_{\bar{T}} \right) \leq
r_H \left( x|_{T'}, \hat{x}|_{T'} \right) + \varepsilon \right\}$ and
arbitrary coefficients $\alpha_x^{test} \in \mathbb{C}$.

\noindent For any fixed $test = (T, T')$, and for any fixed $x|_T \in
\bin^{\alpha m}$ with $\delta = r_H \left( x|_{T'}, \hat{x}|_{T'} \right) \leq
\frac{1}{2}$, let $\ket{\psi_{AE}}$ be the state to which
$\ket{\tilde{\varphi}_{AE}^{test}}$ collapses when for every $i \in T$
subsystem $A_i$ is measured in basis $\hat{\theta}_i$ and $x_i$ is observed,
where we understand $A$ in $\ket{\psi_{AE}}$ to be restricted to the registers
$A_i$ with $i \in \bar{T}$. Finally, let $\sigma_E =
\textnormal{tr}_A(\ketbra{\psi_{AE}}{\psi_{AE}})$ and let the random variable
$X$ describe the outcome when measuring the remaining $n = (1-\alpha)m$
subsystems of $A$ in basis $\theta|_{\bar{T}} \in \bases{}^n$. Then, for any
subset $I \subseteq \left\{ 1, \dots, n \right\}$ and any $x|_I$,
\[\entro{\infty}{X|_I| X|_{\bar{I}} = x|_{\bar{I}}} \geq \hamming{\theta|_I,
\hat{\theta}|_I} - h(\delta+\varepsilon)n \textnormal{\quad and \quad}
\entro{0}{\sigma_{E}} \leq h(\delta + \varepsilon)n.\]
\end{lemma}

With this in hand, we wish to bound the smooth max-entropy of the state when
passing the parameter estimation test (\texttt{flow-one} and \texttt{flow-two}
in our protocol).
\begin{proposition}
\label{prop:proof_1_max_entropy}
For any $\varepsilon > 0$ such that $\varepsilon^2 < \pr{K_{P_1} = K_{P_2}
\neq \perp}$, the following holds:
\[\smaxentro{X_WW|Z'E} \leq h(\tau + \varepsilon)n. \]
\end{proposition}
\begin{proof}
We define the event $\mathcal{E} = r_H\left( X_W, \hat{X}_W \right) \geq
\tau$. Hoeffding's inequality \cite{hoeffding1963probability} gives an upper
bound on the probability of the unlikely coincidence where the parameter
estimation test passes with threshold $\tau$ but the fraction of errors
between $X_W$ and $\hat{X}_W$ exceeds the threshold $\tau$ by a constant
amount, that is,
\[\pr{K_{P_1} = K_{P_2} \neq \perp \wedge \mathcal{E}} \leq \varepsilon^2.\]
By mean of smoothing, we remove the above unlikely event from our state
$\rho_{X_WWZ'E}$:
\[\smaxentro{X_WW|Z'E} \leq \entro{0}{X_WW|Z'E} \leq h(\tau+\varepsilon)n.
\]
\end{proof}

By applying the uncertainty relation from Theorem
\ref{theorem:uncertainty_tradeoff}, we get the lower bound on the smooth
min-entropy of $P_1$'s measurement outcomes.
\begin{proposition}
\label{prop:proof_1_min_entropy}
For any $\varepsilon > 0$ such that $\varepsilon^2 < \pr{K_{P_1} = K_{P_2}
\neq \perp}$, the following holds:
\[\sminentro{X_WW|ZE} \geq nq - \frac{n}{2},\] where we introduced the
shorthand $q = \log{\frac{1}{\bar{c}}} - h \left( \tau +
\varepsilon \right)$, for some $\bar{c} \in \left( 0, 1 \right) $.
\end{proposition}
\begin{proof}
Combining Proposition \ref{prop:proof_1_max_entropy} and the Uncertainty
Relation Theorem \ref{theorem:uncertainty_tradeoff} yields
\[\sminentro{X_WW|Z'E} \geq nq.\] Finally, we show the lower bound of
the smooth min-entropy after \texttt{flow-five}:
\[nq \leq \sminentro{X_WW|Z'E} \leq \sminentro{X_WW|Z'SJE}
+ \frac{n}{2},\] which follows by the chain rule \ref{lemma:chain_rule}, and
the fact that $\size{S} = \frac{n}{2}$ and $J$ is uniformly distributed.
We conclude the proof by summarizing $Z = \left(Z', S, J \right) $.
\end{proof}

With all these necessary technical ingredients, we establish the secrecy of
the key and finish the proof of the second claim.
\begin{proof}[Proof of Claim {\ref{claim:proof_1_secrecy}}]
By the triangle inequality, we have that
\[\dis{\rho_{K_{P_1}WZE}} \leq \dis{\rho_{K_{P_1}WZ'E}} +
\dis{\rho_{K_{P_1}WSJE}}.\]

The Privacy Amplification Theorem \ref{theorem:privacy_amplification} and
Private Error Correction Theorem \ref{theorem:private_ecc} applied with the
bound given in Proposition \ref{prop:proof_1_min_entropy} then immediately
yields the desired inequality.
\end{proof}

We complete the proof of the case when neither of the parties is corrupted by
proving the following claim (this is the difference of a PAKE protocol from
the standard key exchange's proof).
\begin{claim}
\label{claim:proof_1_independence} $W$ is independent of the adversary's
quantum system:
\[\markovstate{W}{Z}{E} = \rho_{WZE}.\]
\end{claim}
\begin{proof}
Because $W$ is independent of $Z$, it is sufficient to show that whether the
protocol was aborted or completed gives no additional information on $W$: By
using Lemma \ref{lemma:proof_1_bounding_entropy}, we are close to the case
where for \emph{any} choice of $T$ and $T'$, and for \emph{any} outcome $x_T$
when measuring $\ket{\Psi}$ in basis $\theta|_{T_2}$ and
$\hat{\theta}|_{T_1}$, the relative error $r_H \left( x|_{T'}, \hat{x}|_{T'}
\right) $ gives an upper bound (except with a negligible probability) on the
relative error $r_H \left(x|_{\bar{T}}, \hat{x}|_{\bar{T}} \right) $ obtained
by measuring the remaining subsystems with $i \in \bar{T}$. In
particular, we have
\[r_H \left( x|_{\bar{T}}, \hat{x}|_{\bar{T}} \right)  \leq \frac{1}{2}r_H
\left( x|_{T'}, \hat{x}|_{T'} \right). \]

Furthermore, since the set $I_w = \left\{ i \in \bar{T} \mid \theta_i =
\hat{\theta_i} \right\}$ is a subset of $\bar{T}$ of essentially half the
size, $r_H \left( x|_{I_w}, \hat{x}|_{I_w} \right) \leq 2\cdot r_H \left(
x|_{\bar{T}}, \hat{x}|_{\bar{T}} \right)$ holds with overwhelming
probability. Also note that $r_H \left( x|_{\bar{T}},
\hat{x}|_{\bar{T}} \right)$ does not depend on $pw$. We can now do the case
distinction:
\begin{itemize}[align=left]
\item[\textit{Case 1:}] If $r_H \left( x|_{\bar{T}}, \hat{x}|_{\bar{T}}
\right) \leq \frac{\tau}{2}$ then $x|_{I_w}$ and $\hat{x}|_{I_w}$ differ in at
most a $\tau$-fraction of their positions, and thus the server $P_2$ correctly
recovers $\tilde{x}|_{I_w} = x|_{I_w}$, no matter what $pw$ is.
\item[\textit{Case 2:}] If $r_H \left( x|_{\bar{T}}, \hat{x}|_{\bar{T}}
\right) \geq \frac{\tau}{2}$ then $r_H \left( x|_{T'}, \hat{x}|_{T'} \right)
\geq \tau$. Hence, the protocol always aborts during either \texttt{flow-one}
or \texttt{flow-two}.
\end{itemize}
We have shown that for both cases the value of $W$ only depends on the
adversary's behavior, which proves the claim.
\end{proof}


\subsection{Simulating the case when the Client is Corrupted}
\label{proof:corrupted_client}
\paragraph{Description of the Simulator. } Recall that in general the
simulator $\adv{S}$ needs to extract the corrupted party’s input in order to
send it to the trusted party, and needs to simulate its view so that its
output corresponds to the output received back from the trusted party. The
simulator $\adv{S}$ works as follows.

\begin{enumerate}
\item When initialized with security parameter $\lambda$, $\adv{S}$ first runs
the key-generation algorithms of the dual-mode commitment scheme $\com$
three times, and obtains the key pairs $\left( ck_{\fsub{H}}, tk \right)
\leftarrow \com.\KeyGen_\fsub{H}(1^\lambda) $, $ \left(
ck^\prime_{\fsub{B}}, xk\right) \leftarrow \com.\KeyGen_\fsub{B}(1^\lambda)$
and $ \left(ck^\prime_{\fsub{H}}, tk^\prime \right) \leftarrow
\com.\KeyGen_\fsub{H}(1^\lambda)$.
\item The simulator also chooses a ``dummy password'' $pw'$ at random for the
simulated copy of the honest party $P_2$.
\item $\adv{S}$ initializes the real-world adversary $\adv{A}$, giving it the
pair $(ck, ck') = (ck_{\fsub{H}}, ck^\prime_{\fsub{B}})$ as the common
reference string. Thereafter, $\adv{S}$ interacts with the ideal functionality
$\ideal{pwKE}$ and its subroutine $\adv{A}$. (Essentially, $\adv{S}$ uses $ck
= ck_{\fsub{H}}$ for its commitments and $\adv{A}$ uses $ck' =
ck^\prime_{\fsub{B}}$ for its commitments. See details in Step
\ref{enum:step_corrupted_p1_2} right below.)
\item This interaction is implemented by the simulator $\adv{S}$ just
following the protocol $\Pi$ on behalf of the honest party, except for the
following modifications:
\begin{enumerate}
\item $\adv{S}$ measures the received qubits only when needed. In the
commitment phase, it simply commits to $\hat{x} = 0^k$, and commits to
$\hat{\theta}$ honestly. It measures the qubits within set $T_1$ upon
receiving $\mathtt{flow-one}_2$, then, in the open phase, to open as an
arbitrary $(\hat{x}|_{T_1}, \hat{\theta}|_{T_1})$, the simulator just uses the
trapdoor information $tk$ to equivocate.
\item \label{enum:step_corrupted_p1} For each commitment $(c_i^0, c_i^1)$
received from the simulated adversary $\adv{A}$, $\adv{S}$ uses the trapdoor
information $xk$ to extract $\adv{A}$'s committed values $x$ and $\theta$. It
then uses the commitment key $ck^\prime_{\fsub{H}}$ to re-commit these values.
\item $\adv{S}$ measures all the remaining qubits in $\adv{A}$'s basis
$\theta$. However, it still verifies only whether $x_i = \hat{x}_i$ for those
$i \in T_2^\prime = \left\{ i \in T_2 \mid \theta_i = \hat{\theta}_i
\right\}$.
\item It sends a random value $\hat{\varphi} \in \bin^n$ to $\adv{A}$.
\item Upon receiving $\varphi$ from $\adv{A}$, it attempts to decode $pw^*$
from $\mathfrak{c}(pw) = \theta|_{\bar{T}} \oplus \varphi$. If this succeeds,
it sets $pw'$ equal to $pw^*$ and uses $pw'$ in a $\textsf{TestPwd}$ query to
$\ideal{pwKE}$. If this is a ``correct guess'', $\adv{S}$ replaces the dummy
password $pw'$ with the correct password $pw^*$, and proceeds with the
simulation.
\end{enumerate}
\item \label{enum:step_corrupted_p1_2} $\adv{S}$ outputs whatever $\adv{A}$
outputs, except it replaces $\adv{A}$'s commitments with its own ones in Step
\ref{enum:step_corrupted_p1}.
\end{enumerate}

\paragraph{Proof of Indistinguishability.} We need to show that the state
output by $\adv{S}$ above is statistically close to the state output by
$\adv{A}$ when executing $\Pi$ with a real $P_2$. First, note that because
$\adv{A}$ is quantum-polynomial-time, it cannot distinguish the commitment key
provided by the simulator and the commitment key in the real world, except
with a negligible probability. If $pw'$ is a correct guess, then the
simulation is perfect and thus the two states are equal. It thus suffices to
argue that the key $\textsf{sk}$ that the Server $P_2$ computes is uniformly
random from the view of $\adv{A}$ for any fixed $pw \neq pw'$.

Now, by re-using Lemma \ref{lemma:proof_1_bounding_entropy} from
\cite[Corollary 4.4]{C:DFLSS09}, we get the common state after the
\texttt{flow-two} is statistically close to a state for which it is guaranteed
that $\entro{\infty}{X|_I} \geq  \hamming{\theta|_I, \hat{\theta}|_I} -
h(\tau + \varepsilon) n$ for any $I \subseteq \left\{ 1, \dots, n \right\}$ and
$\entro{0}{\rho_{ZE}} \leq h(\tau + \varepsilon) n$.
We make a case distinction:
\begin{itemize}[align=left]
\item[\textit{Case 1:}] Decoding of $\theta|_{\bar{T}} \oplus \varphi$
succeeded, \ie{} $\mathfrak{c}(pw^*) = \theta|_{\bar{T}} \oplus \varphi$.
Since the code $\mathfrak{c}$ has a minimum distance $d = \gamma n$, it
follows that $\theta|_{\bar{T}} \oplus \varphi$ is at least $\gamma n$ from
$\mathfrak{c}(pw)$, for any $pw^* \neq pw$.
\item[\textit{Case 2:}] Decoding of $\theta|_{\bar{T}} \oplus \varphi$ failed,
it also follows that $\theta|_{\bar{T}} \oplus \varphi$ is at least $\gamma n$
from $\mathfrak{c}(pw)$, since then $\theta|_{\bar{T}} \oplus \varphi$ is at
least $\gamma n$ from \emph{any} codeword.
\end{itemize}

In both cases, we always have $\hamming{\theta|_{\bar{T}}, \mathfrak{c}(pw)
\oplus \varphi} = \hamming{\theta|_{\bar{T}} \oplus \varphi, \mathfrak{c}(pw)}
\geq \gamma n$. Furthermore, denote $\mathfrak{c}'(pw) = \mathfrak{c}(pw)
\oplus \varphi$, by the random sampling theory, the Hamming distance between
$\theta|_{I_w} $ and $\hat{\theta}|_{I_w} = \mathfrak{c}'(pw) |_{I_w} $ is at
least $\left( \frac{\gamma}{2} - \varepsilon \right)n $, with overwhelming
probability for arbitrary $\varepsilon > 0$. We conclude
that $\entro{\infty}{X_W} \geq \left(\frac{\gamma}{2} - \varepsilon -
h(\tau + \varepsilon)
\right)n$ and $\entro{0}{\rho_{ZE}} \leq h(\tau + \varepsilon) n$.
Hence, the chain rules for smooth min-entropy implies that
\[\sminentro{X_W|ZE} \geq \entro{\infty}{X_W} - \entro{0}{\rho_{ZE}} \geq
\left(\frac{\gamma}{2} - \varepsilon - 2h(\tau + \varepsilon) \right)n.\]
We can complete the proof as follows.
\begin{align*}
\dis{\rho_{K_{P_2}WW'Z'E|W'\neq W}} &= \dis{\markovstate{W}{W'}{Z'E|W'\neq W}}
\\
&\leq \varepsilon_{pa} = 2\varepsilon + 2^{-\frac{1}{2}\left( \left(\frac{
\gamma}{2} - \varepsilon - 2h(\tau + \varepsilon) \right)n-\lambda \right) },
\end{align*}
and
\begin{align*}
\dis{\rho_{K_{P_2}WW'SJE|W'\neq W}}
&= \dis{\markovstate{W}{W'}{SJE|W'\neq W}} \\
&\leq \varepsilon_{ec} = \delta \times 2^{-\frac{1}{2}\left(\frac{\gamma}{2} -
\varepsilon - 2h(\tau + \varepsilon) -\frac{1}{2}\right)n}\\
&\leq 2^{\frac{-\beta n}{4}} \times 2^{-\frac{1}{2}\left(\frac{\gamma}{2} -
\varepsilon - 2h(\tau + \varepsilon) -\frac{1}{2}\right)n} \\
&= 2^{-\frac{1}{2}\left(\frac{\beta}{2} + \frac{\gamma}{2} -
\varepsilon - 2h(\tau + \varepsilon) -\frac{1}{2}\right)n},
\end{align*}
where both exact equalities come from the independency of $W$, which, when
conditioned on $W' \neq W$, translates to independency given $W'$, and the
inequalities follow by privacy amplification (Theorem
\ref{theorem:privacy_amplification}) and private error correction (Theorem
\ref{theorem:private_ecc}), respectively.
By the choice of parameters, we have that both two bounds are
negligible in $n$.

Overall, by the triangle inequality,
the claim follows with $\mu = \varepsilon_{ec} + \varepsilon_{pa}$, that is,
\begin{align*}
\td{\rho_{K_{P_2}WW'ZE|W'\neq W}}{\markovstate{W}{W'}{ZE|W'\neq W}} = &\enskip
\dis{\rho_{K_{P_2}WW'ZE|W'\neq W}} \\
\leq &\enskip \dis{\rho_{K_{P_2}WW'Z'E|W'\neq W}} + \\
& \enskip \dis{\rho_{K_{P_2}WW'SJE|W'\neq W}} \\
\leq & \enskip \varepsilon_{pa} + \varepsilon_{ec}.
\end{align*}


\subsection{Simulating the case when the Server is Corrupted}
\label{proof:corrupted_server}
\paragraph{Description of the Simulator.} The simulator for the corrupted
Server follows the same strategy as the case of the corrupted Client.
Formally, the simulator $\adv{S}$ works as follows.

\begin{enumerate}
\item When initialized with security parameter $\lambda$, $\adv{S}$ first runs
the key-generation algorithms of the dual-mode commitment scheme $\com$
three times, and obtains the key pairs $\left( ck_{\fsub{H}}, tk \right)
\leftarrow \com.\KeyGen_\fsub{H}(1^\lambda) $, $ \left( ck_{\fsub{B}},
xk\right) \leftarrow \com.\KeyGen_\fsub{B}(1^\lambda)$ and $
\left(ck^\prime_{\fsub{H}}, tk^\prime \right) \leftarrow
\com.\KeyGen_\fsub{H}(1^\lambda)$.
\item The simulator also chooses a ``dummy password'' $pw'$ at random for the
simulated copy of the honest party $P_1$.


\item $\adv{S}$ initializes the real-world adversary $\adv{A}$, giving it the
pair $(ck, ck') = (ck_{\fsub{B}}, ck^\prime_{\fsub{H}})$ as the common
reference string. Thereafter, $\adv{S}$ interacts with the ideal functionality
$\ideal{pwKE}$ and its subroutine $\adv{A}$. (Essentially, $\adv{S}$ uses
$ck^\prime = ck^\prime_{\fsub{H}}$ for its commitments and $\adv{A}$ uses $ck
= ck_{\fsub{B}}$ for its commitments. See details in Step
\ref{enum:step_corrupted_p2_2} right below.)
\item This interaction is implemented by the simulator $\adv{S}$ just
following the protocol $\Pi$ on behalf of the honest party, except for the
following modifications:
\begin{enumerate}
\item $\adv{S}$ runs an equivalent EPR-pair version of the protocol, where it
creates $k$ EPR pairs $\left( \ket{00} + \ket{11} \right)/\sqrt{2}$, sends one
qubit in each pair to the adversary $\adv{A}$ and keeps the others in the
register $A$.
\item \label{enum:step_corrupted_p2} For each commitment $(c_i^0, c_i^1)$
received from the simulated adversary $\adv{A}$, $\adv{S}$ uses the trapdoor
information $xk$ to extract $\adv{A}$'s committed values $\hat{x}$ and
$\hat{\theta}$. It then uses the commitment key $ck_{\fsub{H}}$ to re-commit
these values.
\item Instead of measuring its qubits in $T_1$ in its basis $\theta|_{T_1}$,
$\adv{S}$ measures them in $\adv{A}$'s basis $\hat{\theta}|_{T_1}$. However,
it still verifies only whether $\hat{x}_i = x_i$ for those $i \in T_1^\prime =
\left\{ i \in T_1 \mid \theta_i = \hat{\theta}_i \right\}$. Because the
positions $i \in T_1$ with $\theta_i \neq \hat{\theta}_i$ are not used in the
protocol at all, this change has no effect.
\item In the commitment phase, $\adv{S}$ commits to $x = 0^{\left( 1-\alpha
\right)k }$, and commits to $\theta$ honestly. It measures the qubits within
set $T_2$ upon receiving $\mathtt{flow-two}_2$, then, in the open phase, to
open as an arbitrary $(x|_{T_2}, \theta|_{T_2})$, the simulator just uses the
trapdoor information $tk'$ to equivocate.
\item $\adv{S}$ measures all the remaining qubits in $\adv{A}$'s basis
$\hat{\theta}|_{\bar{T}}$ after the \texttt{flow-two}.
\item Upon receiving $\hat{\varphi}$ from $\adv{A}$, it attempts to decode
$pw^*$ from $\mathfrak{c}(pw) = \hat{\theta}|_{\bar{T}} \oplus \hat{\varphi}$.
If this succeeds, it sets $pw'$ equal to $pw^*$ and uses $pw'$ in a
$\textsf{TestPwd}$ query to $\ideal{pwKE}$. If this is a ``correct guess'',
$\adv{S}$ replaces the dummy password $pw'$ with the correct password $pw^*$.
\item It sends a random value $\varphi \in \bin^n$ to $\adv{A}$, and proceeds
with the simulation.
\end{enumerate}
\item \label{enum:step_corrupted_p2_2} $\adv{S}$ outputs whatever $\adv{A}$
outputs, except it replaces $\adv{A}$'s commitments with its own ones in Step
\ref{enum:step_corrupted_p2}.
\end{enumerate}

\paragraph{Proof of Indistinguishability.} Similarly to the proof of
indistinguishability of the case when the client is corrupted, we have that if
$pw'$ is a correct guess, then the simulation is perfect. Otherwise, by using
the same argument, we have
\[\sminentro{X_W|ZE} \geq \left(\frac{\gamma}{2} - \varepsilon -
2h(\tau + \varepsilon)
\right)n.\]

Here, we upper bound $\td{\rho_{K_{P_1}WW'ZE|W'\neq
W}}{\markovstate{W}{W'}{ZE|W'\neq W}}$ only by privacy amplification. In
particular, we have
\begin{align*}
\rho_{K_{P_1}WW'ZE|W'\neq W} & \approx_\mu \markovstate{W}{W'}{ZE|W'\neq W} \\
& = \frac{1}{2^\lambda}\mathbbm{1} \otimes \markovstate{W}{W'}{ZE|W'\neq W},
\end{align*}
where the approximation follows from privacy amplification, and the exact
equality comes from the independency of $W$, which, when conditioned on $W'
\neq W$, translates to independency given $W'$. The claim follows with $\mu =
2\varepsilon + 2^{-\frac{1}{2}\left( \left(\frac{\gamma}{2} - \varepsilon -
2h(\tau + \varepsilon) \right)n-\lambda \right)}$.

\end{document}